\documentclass[11pt,reqno]{amsart}
\usepackage[letterpaper,margin=1in,heightrounded]{geometry}

\usepackage[T1]{fontenc}
\usepackage{lmodern}
\usepackage{amsmath}
\usepackage{amssymb}
\usepackage{mathtools}
\usepackage{booktabs}
\usepackage{array}
\usepackage{enumitem}
\usepackage{needspace}
\ifdefined\SubmissionReviewFormat
  \usepackage{setspace}
\fi
\usepackage{microtype}
\usepackage{hyperref}
\usepackage{url}
\usepackage{xcolor}
\usepackage{placeins}
\usepackage{tikz}
\usepackage{tikz-cd}
\usepackage{etoolbox}
\AtBeginEnvironment{thebibliography}{\small\setlength{\itemsep}{0pt}}
\usepackage{longtable}
\usetikzlibrary{arrows.meta,calc}
\definecolor{pathblue}{RGB}{30,84,139}
\definecolor{pathorange}{RGB}{177,83,23}
\definecolor{pathteal}{RGB}{28,117,108}

\hypersetup{
  hidelinks,
  pdftitle={Topological Computational Paths: Presentation Invariance and Observable Quotients},
  pdfauthor={Arthur Freitas Ramos, Ruy J. G. B. de Queiroz, Anjolina Grisi de Oliveira, and Tiago M. L. de Veras},
  pdfsubject={Topological semantics for scoped computational-path rewrite presentations},
  pdfkeywords={computational paths, rewriting, fundamental groupoid, quotient topology, topological groupoid}
}

\allowdisplaybreaks

\theoremstyle{definition}
\newtheorem{definition}{Definition}[section]
\newtheorem{example}[definition]{Example}

\theoremstyle{plain}
\newtheorem{theorem}[definition]{Theorem}
\newtheorem{lemma}[definition]{Lemma}
\newtheorem{proposition}[definition]{Proposition}
\newtheorem{corollary}[definition]{Corollary}
\theoremstyle{remark}
\newtheorem{remark}[definition]{Remark}

\newcommand{\Trace}{\mathsf{Trace}}

\newcommand{\RwEq}{\mathrel{\simeq_{\!\mathcal P}}}

\newcommand{\Fin}{\mathrm{fin}}
\newcommand{\Ord}{\mathrm{ord}}
\newcommand{\Top}{\mathsf{Top}}
\newcommand{\Z}{\mathbb Z}
\newcommand{\N}{\mathbb N}
\newcommand{\Circle}{\mathbb R/\mathbb Z}
\newcommand{\LeanName}[1]{\path{#1}}

\title[Topological Computational Paths]
  {Topological Computational Paths: Presentation Invariance and Observable Quotients}

\author{Arthur Freitas Ramos}

\author{Ruy J. G. B. de Queiroz}

\author{Anjolina Grisi de Oliveira}

\author{Tiago M. L. de Veras}


\newcommand{\FigureAltText}[1]{\par\smallskip{\small\textbf{Alt text:} #1\par}}
\newcommand{\MainFigure}[2]{%
  \begingroup\let\SavedFigureCaption\caption
  \renewcommand{\caption}[1]{\SavedFigureCaption{##1}\FigureAltText{#2}}%
  \input{figures/#1}\endgroup}

\date{3 November AD 2026}

\begin{document}

\raggedbottom

\begin{abstract}
A computational path records the steps used to establish an equality.
Topologies on computational equality witnesses can depend on the chosen
generators and execution schedule, even when their algebraic and
homotopical meanings agree. We study continuous presentations of finite
computational paths whose declared rewrites are realized by endpoint-fixed
homotopies. We compare a topology retaining the full signed trace with an
observable topology recording endpoints, length, and exact path realization.

Continuous finite-word interpretations that are inverse up to scoped
rewriting preserve trace-sensitive quotient arrows and their final
composable-pair spaces. Observable quotients behave differently. Over a point
with a nonempty primitive alphabet, we classify their topology: it records
length parity when all named relations preserve parity, and is otherwise
indiscrete. Adjoining a conservative abbreviation can therefore turn a
discrete two-element quotient into an indiscrete space, even though ordinary
multiplication remains continuous. A circle presentation also gives nonhomeomorphic
quotients under binary and equal-slot timing. Geometric completeness alone
does not identify the trace-sensitive and observable topologies.

Multiplication is continuous on the quotient of composable representatives.
We characterize when this topology agrees with the ordinary pullback topology and give an
open-projection criterion. Universal presentations recover classical
fundamental-groupoid topology. Normal forms, primitive interchange, and a
Cayley covering prove completeness for circle, torus, product, and rose
presentations. A literal-word correspondence relates the semantics to the
reflexivity, symmetry, and transitivity fragment of equality derivations.
Selected results are formalized in Lean, with the realization models and
proof coverage stated explicitly.
\end{abstract}

\maketitle

\begin{center}
\footnotesize
\ifdefined\SubmissionReviewFormat\setstretch{1}\fi
Arthur Freitas Ramos: Microsoft, USA; \texttt{arfreita@microsoft.com}\\
Ruy J. G. B. de Queiroz and Anjolina Grisi de Oliveira: Centro de
Inform\'atica, Universidade Federal de Pernambuco, Brazil;\\
\texttt{ruy@cin.ufpe.br}, \texttt{ago@cin.ufpe.br}\\
Tiago M. L. de Veras: Departamento de Matem\'atica, Universidade Federal
Rural de Pernambuco, Brazil; \texttt{tiago.veras@ufrpe.br}\\[2pt]
\emph{Corresponding author: Arthur Freitas Ramos}
\end{center}

\section{Introduction}
\label{sec:introduction}

\subsection*{Equality with a recorded derivation}

An equality judgement states that two expressions are equal according to
the rules of a formal system. A \emph{computational path} records the finite labelled derivation used to establish that equality. In the
judgement $a=_p b:A$, the expressions $a$ and $b$ have type $A$, and the
label $p$ records the permitted steps connecting them. Each step specifies
an instance of an equality rule or a given equality; a step inside a larger
expression also records the relevant context. The word ``computational''
refers to this explicit derivational content. Such a path can use an
equation in either direction and need not describe the forward execution
of a program. This account of equality witnesses and its identity-type
interpretation are developed in
\cite{RamosDeQueirozDeOliveira2016propositional,Ramos2017identity}.

For example, let $n$ be a natural number and use the equality
rule $z+0=z$. Two applications give
\[
 (n+0)+0\ \xrightarrow{u}\ n+0\ \xrightarrow{v}\ n.
\]
The first label records the instance with $z=n+0$, and the second records
the instance with $z=n$. Their composite witnesses the equality
$(n+0)+0=n$ while retaining these two steps. If a new primitive label $w$
is introduced as an abbreviation for the whole derivation, declaring it
rewrite-equivalent to $u;v$ gives an equation \emph{between witnesses}.
The endpoints already agree; the declaration specifies how the new
witness relates to the old one. Later we ask when such a change of
presentation also preserves the topology assigned to witness classes.

Three operations organize these derivations. Reflexivity supplies an empty
derivation, symmetry reverses an equality witness, and transitivity joins
two witnesses with a matching intermediate expression. In the usual
computational-path notation they are
\[
 \rho_a:a=a,\qquad
 \sigma(p):b=a,\qquad
 \tau(p,q):a=c
 \quad\text{for }p:a=b\text{ and }q:b=c.
\]
We write these operations as $1_a$, $p^{-1}$, and $p;q$, respectively.
For example, $v^{-1}$ witnesses $n=n+0$. This reversal expresses the
symmetry of equality; it makes no claim that an evaluation procedure can
reconstruct a previous machine state. Structural rewrites express the unit,
associativity, inverse, and cancellation laws. Thus $p;p^{-1}$ rewrites to
$1_a$, and different parenthesizations of a composite can be related.
The classes of witnesses modulo these rewrites form a groupoid: arrows
compose when their endpoints match, and every arrow has an inverse.
Additional equations between parallel derivations must belong to the
specified presentation.

The distinction becomes useful when two reduction orders give the same
result. For a simply typed example, take
$\Gamma=(a:A,\,f:A\to B)$ and put
\[
 \begin{aligned}
 M&=(\lambda x:A.\,f\,x)\bigl((\lambda y:A.\,y)\,a\bigr),\\
 N&=f\bigl((\lambda y:A.\,y)\,a\bigr),
 &P&=(\lambda x:A.\,f\,x)\,a.
 \end{aligned}
\]
Here $\lambda x:A.\,f\,x$ is the function sending $x$ to $f\,x$, and
$\lambda y:A.\,y$ is the identity function. A $\beta$-contraction replaces
an application of such a function by its body with the argument
substituted. All four terms $M,N,P,f\,a$ have type $B$ in $\Gamma$.
Contracting the outer application first or the inner identity application
first gives
\[
 \begin{aligned}
 \Gamma&\vdash M=_r N:B,
 &\Gamma&\vdash N=_s f\,a:B,\\
 \Gamma&\vdash M=_{r'}P:B,
 &\Gamma&\vdash P=_{s'}f\,a:B.
 \end{aligned}
\]
Here $r,s'$ contract the outer abstraction, while $s,r'$ contract the
inner identity application using congruence. The composite witnesses $r;s$
and $r';s'$ retain the two reduction orders. Structural rewrites cancel
$(r;s);(s^{-1};r^{-1})$ to $1_M$, but an equation between $r;s$ and
$r';s'$ requires a witness rewrite supplied by the presentation. The
$\beta$-convertibility of their endpoints does not itself specify that
equation. This leaves two questions: whether the terms are
equal, and which derivations of that equality are identified.

\subsection*{What geometry contributes}

Earlier work develops the weak groupoid of computational paths
\cite{DeVerasRamosDeQueirozDeOliveira2025weak}, applies labelled natural
deduction to fundamental-group calculations
\cite{DeVerasRamosDeQueirozDeOliveira2025topological}, and gives a
book-length account of the calculus
\cite{RamosDeQueirozDeOliveiraGabbay2026}.
We consider its geometric groupoid fragment. Objects are points of a
chosen topological space $X$; in a logical example, selected points
represent the terms under consideration. A presentation assigns a
continuous interval path to each primitive witness. Empty derivations
become constant paths, formal reversals retrace paths, and composites
concatenate them. A path followed by its reverse contracts to a constant
path while its endpoints remain fixed. This gives a geometric explanation
of cancellation and, more generally, of the structural groupoid laws.

Equations between witnesses require a second level of geometric data.
Each named witness equation is assigned an \emph{endpoint-fixed homotopy},
a continuous deformation between the two realized paths that keeps their
common endpoints stationary. For the reduction diamond above, place
$M,N,P,f\,a$ at the four corners of a filled square and realize the four
contractions by its oriented edges. The two composites follow different
boundary routes between opposite corners. The interior of the square
allows one route to deform into the other. Section~\ref{sec:typed-diamond}
uses this realization and proves completeness on the four displayed
objects after declaring the diamond equation.

The choice of geometry matters: the boundary of the square alone does not
permit that deformation. The choice of witness equations matters as well:
even in the filled square, omitting the diamond equation leaves the two
composites distinct under the structural rules. We therefore keep the
logical relation and its geometric interpretation separate. \emph{Scoped
equality} is generated by the named equations and structural groupoid
laws. \emph{Soundness} says that scoped equality produces a homotopy.
\emph{Completeness} asks whether any two traces whose realized paths are
homotopic with fixed endpoints are scoped-equal. It is a property to prove
for a chosen presentation, rather than an extra rewrite available in
every presentation. Geometry can thus supply both semantic witnesses for
the declared rules and invariants that distinguish derivations when no
appropriate homotopy exists.

For a concrete invariant, interpret a primitive loop $e:x\to x$ as one
positively oriented turn around a circle. The composite $e;e^{-1}$ contracts, whereas $e$
has winding number $1$ and cannot contract. Soundness therefore rules out
a scoped derivation of $e\RwEq1_x$. For words in $e$ and $e^{-1}$, the
signed traversal count gives the winding number, which remains unchanged
under scoped rewriting.
In the one-generator circle presentation, completeness proves that this
invariant is exact: two represented loops are scoped-equal precisely when
their winding numbers agree (Theorem~\ref{thm:circle-winding}). Geometry
therefore helps both to disprove proposed witness equations and to certify
that the declared rules identify precisely the represented homotopy classes.

\subsection*{Which information should topology retain?}

The existence of a homotopy is only part of the geometric information.
Primitive witnesses may vary in continuous families: for example, a
once-around circle loop varies continuously as its base point moves
around the circle. Their realized paths have neighborhoods in the
compact-open path space. A topology on
the witness space specifies which such variations are continuous.
After identifying scoped-equal witnesses, the quotient topology declares
a set of classes open precisely when its full inverse image among
representatives is open. Two presentations can therefore give the same
abstract groupoid while assigning different topologies to its arrows.
Computing the underlying group or groupoid does not settle this question.

We compare two explicit choices of information. A finite trace is a
composable signed word: each letter names a primitive witness and records
whether it is traversed forwards or backwards. We give its letters equal
time slots in the realized interval path. The \emph{trace-sensitive}
topology retains the complete signed word, using the given topology on
primitive labels. The \emph{observable} topology records only the
endpoints, the discrete word length, and the exact realized path.
Distinct traces with the same observable data are then topologically
indistinguishable, even before scoped rewriting identifies any witnesses.
The exact path includes its parametrization, so changing the execution
schedule can change this information while preserving homotopy classes.
A primitive realized by a constant path still has length one; only the
empty word has length zero. These witnesses need not be scoped-equal.
Geometric realization and the recorded derivation therefore carry
different information even in this simple case.

A coherent representative may additionally carry a path homotopic to the
trace's prescribed realization. Forgetting this additional path induces
a homeomorphism on scoped quotients for the matching trace topology.
The comparison therefore depends on the choice between full trace data
and observable data, rather than on this auxiliary representative.
Three changes motivate our main question: changing parentheses, replacing
a derived witness by a primitive abbreviation, and changing the time
assigned to elementary steps. Explicit translations can control their
algebraic effects. When do those translations also preserve the topology
of the quotient witness space?

\Needspace{4\baselineskip}
\subsection*{Results and relation to earlier work}

\begin{enumerate}[label=(\roman*),leftmargin=2.2em]
\item \emph{Full-word invariance.}
  Continuous primitive word codes, preservation of named rules, and scoped
  primitive roundtrips induce homeomorphisms of trace-sensitive global
  quotients and final composable-pair spaces
  (Theorem~\ref{thm:flat-word-invariance}). The codes may have variable
  lengths, including zero. Literal-word correspondence makes the change
  from parenthesized groupoid terms precise.
\item \emph{Observable dependence.}
  Over a point with a nonempty primitive alphabet, we determine the
  observable topology exactly from parity of the named relations
  (Theorem~\ref{new-thm:point-parity}). A conservative abbreviation changes
  a discrete two-element quotient into an indiscrete space, proving
  intrinsic nonhomeomorphism
  (Corollary~\ref{new-cor:intrinsic-abbreviation-change}). Both presentations
  still have continuous ordinary multiplication. A three-label circle
  presentation separates binary and equal-slot timing
  (Proposition~\ref{prop:binary-timing-collision}), while a complete
  harmonic-archipelago presentation shows that completeness alone does not
  equate the two quotient topologies
  (Proposition~\ref{prop:complete-trace-separation}).
\item \emph{The composition domain.}
  Multiplication is always continuous on the quotient of composable
  representatives. This final topology must be compared with the ordinary
  topology inherited from the product of quotient arrow spaces: taking
  quotients and forming products need not commute. We characterize
  agreement and prove an open-arrow-projection criterion. The universal
  presentation connects this criterion to the classical topology of
  fundamental groupoids; the Hawaiian earring exhibits failure of
  ordinary multiplication continuity.
\end{enumerate}

The algebraic presentation mechanism and its universal constructions have
classical precedents in topological groupoids \cite{BrownHardy76,HardyMorris77}.
For discrete object spaces, Brazas builds word categories from topological
graphs and distinguishes the word-reduction quotient from its
topological-groupoid reflection \cite[Section~3.2]{Brazas14}.
Brazas and Emery study free quasitopological groups using cross topologies
\cite{BrazasEmery23}. Tietze transformations provide related algebraic
and coherence results for presentations
\cite{HenryMimram21,MimramOleon25}. Our question concerns continuity for
the specified full-word and geometric observation topologies. The chosen
arrow quotient remains fixed when its multiplication domain is compared
with the ordinary pullback. We use general quotient-space facts to establish
presentation-dependent results on invariance, loss of observable
information, and composition.

Normal forms establish completeness for the circle and torus, primitive
interchange preserves based completeness under products, and a geometric
Cayley covering classifies all continuous based loops of the two-circle
rose. These are applications of the semantics to classical geometry.
The Lean development checks selected results for explicitly distinguished
binary and equal-slot models. The verification map states the declarations,
assumptions, and proof coverage used by the manuscript.

\section{Geometric traces and scoped presentations}
\label{sec:traces}

Throughout, $X$ is a topological space and $I=[0,1]$.  Give the interval-path
space $X^I$ the compact-open topology.  Let $E$ be a topological space of
oriented primitive steps, with continuous endpoint maps $s,t:E\to X$ and a
continuous realization map
\[
  \rho:E\longrightarrow X^I,
  \qquad \rho(e)(0)=s(e),\quad \rho(e)(1)=t(e).
\]
The continuity requirement is part of the presentation data.  Let $E^+$ be a
topological copy of $E$, let $E^-$ be a formal inverse copy, and write
\(\widetilde E=E^+\sqcup E^-\).  On $E^+$ we use the given endpoint and
realization maps; on $E^-$ the endpoint maps are interchanged and the
realization is reversed.  Thus there are continuous extensions
\[
  \widetilde s,\widetilde t:\widetilde E\longrightarrow X,
  \qquad
  \widetilde\rho:\widetilde E\longrightarrow X^I,
\]
whose restrictions to $E^+$ are $s,t,\rho$ and whose restrictions to the
inverse copy (under $E^-\cong E$) are $t,s,e\mapsto\rho(e)^{-1}$,
respectively.

Figure~\ref{fig:oriented-steps} shows composition of labelled witnesses and
the soundness implication from scoped rewriting to geometric homotopy.  The inverse copy does not require an
orientation on the topological space $E$: only the source, target, and
direction of the represented path are reversed.

\begingroup\captionindent=0pt
\setlength{\intextsep}{6pt}
\MainFigure{oriented-steps}{Labelled steps from a through b to c compose to p=r;s. A parallel witness u is connected to p only when the scoped rules derive their equality, which implies endpoint-fixed homotopy of their realizations. Formal reversal exchanges endpoints and reverses time.}
\endgroup

For composable paths $\alpha:a\to b$ and $\beta:b\to c$ and
weights $m,n\in\N$, write $\alpha\star_{m,n}\beta$ for
length-weighted concatenation. A zero-weight factor is required to have
equal endpoints: $m=0$ requires $a=b$, and $n=0$ requires $b=c$.
This convention applies to every weighted factor below. When $m,n>0$,
concatenation is given by
\[
 (\alpha\star_{m,n}\beta)(t)=
 \begin{cases}
   \alpha\!\left(\dfrac{m+n}{m}t\right),
     &0\leq t\leq\dfrac{m}{m+n},\\[7pt]
   \beta\!\left(\dfrac{(m+n)t-m}{n}\right),
     &\dfrac{m}{m+n}\leq t\leq 1.
 \end{cases}
\]
Use the natural zero-length conventions
$\alpha\star_{0,n}\beta=\beta$ and
$\alpha\star_{m,0}\beta=\alpha$; when $m=n=0$, use the constant path at
the common endpoint.  Thus the two input paths occupy $m$ and $n$ equal
word slots, respectively.  For composable paths $\alpha$, $\beta$, and
$\gamma$ occupying $m$, $n$, and $k$ word slots, respectively, this gives the
following literal identities, including the zero-length cases under the stated
conventions:
\[
 (\alpha\star_{m,n}\beta)\star_{m+n,k}\gamma
 =\alpha\star_{m,n+k}(\beta\star_{n,k}\gamma),
 \qquad
 (\alpha\star_{m,n}\beta)^{-1}
 =\beta^{-1}\star_{n,m}\alpha^{-1}.
\]

For example, $\alpha\star_{2,1}\beta$ follows $\alpha$ until
time $2/3$ and $\beta$ during the remaining third. The weights count
word slots, independently of parentheses and homotopy witnesses.

\begin{definition}[geometric trace]
For $n\geq 1$, let $W_n\subseteq\widetilde E^n$ be the subspace of composable
words $(e_1,\ldots,e_n)$ satisfying
\(\widetilde t(e_i)=\widetilde s(e_{i+1})\) for every $i$.  Put $W_0=X$, representing empty traces,
with both endpoint maps equal to the identity on $X$, and define the
topological trace carrier by the disjoint union
\[
  \mathsf{Tr}_{\mathcal P}=\coprod_{n\in\N}W_n.
\]
Give $\N$ the discrete topology.  The length map is the coproduct index and
is continuous.  For $p=(e_1,\ldots,e_n)\in W_n$ with
$n\geq1$, define its realization by equal word slots,
\[
 |p|(t)=\widetilde\rho(e_i)(nt-i+1)
 \quad\text{for }1\leq i\leq n,\quad
 \frac{i-1}{n}\leq t\leq\frac{i}{n}.
\]
For $p\in W_0=X$, let $|p|$ be the constant path at its point.  We write
\(\Trace(a,b)\) for the set of words with endpoints $a,b$.  Empty words, word
concatenation, and reversal are the identity, composition, and inverse
operations on traces.  In particular, for $p\in W_m$ and
$q\in W_n$,
\[
 |p;q|=|p|\star_{m,n}|q|.
\]
Representing traces as flat words makes the unit, associativity,
double-inverse, and reversal-of-composite laws literal equalities.
Their equal-slot realizations satisfy the same equalities.
Inverse cancellation is different: a word followed by its reversal need
not be the empty word, and its realization contracts by a homotopy.

\end{definition}

\begin{lemma}[continuity of trace realization]
\label{lem:trace-realization-continuous}
The equal-slot realization map
\[
  |\cdot|:\mathsf{Tr}_{\mathcal P}\longrightarrow X^I
\]
is continuous for the coproduct trace topology.
\end{lemma}

\begin{proof}
On $W_0=X$, the constant-path map has continuous adjoint $(a,t)\mapsto a$.
For $n\geq1$, the displayed slot formula is continuous on each closed set
$W_n\times[(i-1)/n,i/n]$: the coordinate map to $\widetilde E$, the
primitive realization, and compact-open evaluation are continuous.
Adjacent formulas agree at $t=i/n$ because
$\widetilde t(e_i)=\widetilde s(e_{i+1})$.
Finite closed pasting gives a continuous adjoint $W_n\times I\to X$.
Since $I$ is compact Hausdorff, the compact-open exponential law gives
continuity into $X^I$ on every stratum. The coproduct property assembles
these maps.
\end{proof}

The compact-open topology also gives continuous constant-path and reversal
maps on $X^I$.  On each fixed pair of length components, the weighted
operation $\star_{m,n}$ is continuous on the corresponding endpoint
pullback; the length coordinates are discrete, so these operations assemble
continuously over the coproduct.  Word concatenation and reversal are
continuous on the corresponding endpoint pullbacks of the $W_n$.  Together
with Lemma~\ref{lem:trace-realization-continuous}, these are the continuity
facts used below.  Nothing is assumed about products of quotient maps.

Let
\[
  \mathsf{Coh}_{\mathcal P}(p,\gamma)\;:\Longleftrightarrow\;
  \text{there exists an endpoint-fixed homotopy from $\gamma$ to $|p|$}.
\]
Only the existence of such a homotopy is required; the homotopy itself is
not an entry of the pair.  The coherent carrier is the set
\[
  T=\{(p,\gamma)\in\mathsf{Tr}_{\mathcal P}\times X^I:
      \gamma(0)=s(p),\ \gamma(1)=t(p),\ \mathsf{Coh}_{\mathcal P}(p,\gamma)\}.
\]
The trace $p$ carries the scoped computation; $\gamma$ is a chosen geometric
representative observed by the topology. Coherence does not turn an ambient
homotopy into a declared rewrite. The scoped quotient forgets $\gamma$ within
each scoped trace class, although it contributes to the raw topology.

\subsection*{Geometric meaning of coherence}
Having the same endpoints is necessary but not sufficient for coherence.
Let $A=(-1,0)$, $B=(0,1)$, and $C=(1,0)$ in $\mathbb R^2$.
Let $\alpha$ traverse the two straight segments $AB$ and $BC$ at equal
speed on the two halves of $I$, and put $\beta(t)=(2t-1,0)$.
In the plane,
\[
 H(u,t)=(1-u)\alpha(t)+u\beta(t)
\]
is a continuous homotopy from $\alpha$ to $\beta$ fixing $A,C$.
Thus, if $|p|=\alpha$, the pair $(p,\beta)$ is coherent.
If instead $X=\mathbb R^2\setminus\{(0,\tfrac13)\}$, both paths still
exist, but the loop $\alpha\star_{1,1}\beta^{-1}$ surrounds the missing
point once.  Its winding is nonzero, so no endpoint-fixed homotopy exists
in $X$ \cite{Hatcher02}.  Figure~\ref{fig:coherence} depicts this change.
The loop between them cannot be filled inside the punctured space.

\MainFigure{coherence}{A triangle contracts to its base segment in the plane; puncturing the triangle leaves a loop that cannot contract while avoiding the puncture.}

More generally, for parallel paths $\alpha,\beta$,
\[
 \alpha\simeq_{\partial I}\beta
 \quad\Longleftrightarrow\quad
 \alpha\star_{1,1}\beta^{-1}\simeq_{\partial I}c_{\alpha(0)}.
\]
This follows from the groupoid laws of endpoint-fixed homotopy classes.
Coherence is a path-space condition. Every named rewrite requires a
semantic homotopy, while a homotopy alone supplies no scoped rewrite.

On this coherent carrier, define two canonical topologies. First, let
\[
  \kappa_{\mathrm{tr}}:T\longrightarrow
    \mathsf{Tr}_{\mathcal P}\times X^I,
  \qquad
  \kappa_{\mathrm{tr}}(p,\gamma)=(p,\gamma).
\]
The \emph{trace-sensitive topology} $T_{\mathrm{tr}}$ is the initial topology
induced by $\kappa_{\mathrm{tr}}$; because $T$ is a subset of
$\mathsf{Tr}_{\mathcal P}\times X^I$, this is just the subspace topology.
In it, the complete finite word is a topological coordinate.  Second, let
\[
  \kappa_{\mathrm{obs}}:T\longrightarrow X\times X\times\N\times X^I\times X^I,
  \qquad
  \kappa_{\mathrm{obs}}(p,\gamma)=(s(p),t(p),\ell(p),|p|,\gamma).
\]
The \emph{observable topology} $T_{\mathrm{obs}}$ is the initial topology
induced by $\kappa_{\mathrm{obs}}$:
it is the coarsest topology making these semantic observables continuous.
Since the endpoint maps, length, and realization are continuous on the trace
carrier, $\kappa_{\mathrm{obs}}$ factors continuously through
$\kappa_{\mathrm{tr}}$.  Consequently the identity
\[
  i_T:T_{\mathrm{tr}}\longrightarrow T_{\mathrm{obs}}
\]
is continuous; that is, the trace-sensitive topology is finer than the
observable one.

Unless indicated otherwise, Sections~\ref{sec:quotient}--\ref{sec:circle-torus}
use $T=T_{\mathrm{obs}}$ and $T^{(2)}=T_{\mathrm{obs}}^{(2)}$.
The basic quotient, geometric-comparison, and universal-completion
constructions apply to either carrier topology. The results of
Section~\ref{sec:functoriality} specify which topology is used.
Completeness concerns the relation alone. Section~\ref{sec:final-ordinary}
checks product-quotient compatibility separately for each topology.

Write $\kappa=\kappa_{\mathrm{obs}}$. Distinct traces with identical
observable codes are topologically indistinguishable in $T_{\mathrm{obs}}$;
the trace-sensitive topology can separate them. Changing only the coherent
path can also change an observable code, although the scoped quotient
will identify all coherent representatives of the same trace.

For $u=(p,\gamma)\in T$, put $s(u)=s(p)$ and $t(u)=t(p)$.
For $*=\mathrm{tr},\mathrm{obs}$, define the common composable set with its
chosen subspace topology by
\[
  T_*^{(2)}=\{(u,v)\in T_*\times T_*:t(u)=s(v)\}.
\]

\begin{lemma}[continuity of coherent operations]\label{lem:coherent-operations-continuous}
For either $*=\mathrm{tr}$ or $*=\mathrm{obs}$, let $\iota:X\to T_*$ send
$a$ to the coherent identity $(1_a,c_a)$, let $\sigma:T_*\to T_*$ send
$(p,\gamma)$ to
$(p^{-1},\gamma^{-1})$, and let
\[
  \mu_*:T_*^{(2)}\longrightarrow T_*,
  \qquad ((p,\gamma),(q,\delta))\longmapsto
    (p;q,\gamma\star_{\ell(p),\ell(q)}\delta),
\]
where the weighted operation uses the discrete trace lengths.  Then all three
maps are continuous.
\end{lemma}

\begin{proof}
For the observable topology, the identity and reversal codes are
$(a,a,0,c_a,c_a)$ and
$(t(p),s(p),\ell(p),|p|^{-1},\gamma^{-1})$.
On the open-and-closed length component $(m,n)$, the multiplication code is
\[
 \bigl(s(p),t(q),m+n,
       |p|\star_{m,n}|q|,\gamma\star_{m,n}\delta\bigr).
\]
All coordinates are continuous by the path operations established above;
the zero-length cases use projections or the constant-path map.
Discrete lengths let these formulas assemble continuously, so the
initial-topology criterion applies. For the trace-sensitive topology, add
the continuous word maps $a\mapsto1_a$, $p\mapsto p^{-1}$ and
$(p,q)\mapsto p;q$ on the composable tuple strata.
The operations preserve coherence: reverse a witnessing homotopy or
concatenate two such homotopies with the fixed weights. When a weight is
zero, both weighted expressions omit that factor; when both vanish, both
are constant at the common endpoint. This does not require the chosen
representative of a zero-length trace to be literally constant.
\end{proof}

\begin{definition}[scoped presentation]
A continuous geometric rewrite presentation $\mathcal P$ on $(X,E)$ consists
of a predicate
\[
  \mathsf{Rule}_{\mathcal P}(p,q),\qquad p,q:\Trace(a,b),
\]
and, for every instance of $\mathsf{Rule}_{\mathcal P}(p,q)$, an endpoint-fixed homotopy
between the realizations of $p$ and $q$.
\end{definition}

Here $\mathsf{Rule}_{\mathcal P}$ specifies the named rewrites; it should
not be confused with the realization comparison $R_{\mathcal P}$ of Section~\ref{sec:comparison}.
The definition asks for a homotopy witness for each named rewrite and imposes
no completeness condition.

\begin{definition}[scoped rewrite equality]\label{def:scoped-rewrite-equality}
For parallel traces, write $p\RwEq q$ for the inductively generated relation
whose constructors are reflexivity, the named generators $\mathsf{Rule}_{\mathcal P}$,
symmetry, transitivity, congruence under concatenation and reversal, and the
following structural groupoid coherences:
\[
\begin{array}{lll}
(1_a);p\RwEq p, & p;(1_b)\RwEq p, & (p;q);r\RwEq p;(q;r),\\
p^{-1};p\RwEq 1_b, & p;p^{-1}\RwEq 1_a, & (p^{-1})^{-1}\RwEq p,\\
(1_a)^{-1}\RwEq 1_a, & (p;q)^{-1}\RwEq q^{-1};p^{-1}. &
\end{array}
\]
Only the rules declared by $\mathcal P$ and these structural constructors are
available.  Literal equality of traces is handled by reflexivity and
substitution; ambient homotopy of their realizations is not an additional
constructor.
\end{definition}

\begin{lemma}[weighted concatenation invariance]
\label{lem:weighted-concat-invariance}
Let $\alpha\simeq_{\partial I}\alpha'$ and
$\beta\simeq_{\partial I}\beta'$ be endpoint-fixed homotopies of composable
paths.  Write $c_x$ for the constant path at $x$, and assume explicitly that
\[
\begin{aligned}
m=0&\Rightarrow \alpha(0)=\alpha(1)\ \text{and}\ \alpha\simeq_{\partial I}c_{\alpha(0)},\\
n=0&\Rightarrow \beta(0)=\beta(1)\ \text{and}\ \beta\simeq_{\partial I}c_{\beta(0)},\\
m'=0&\Rightarrow \alpha'(0)=\alpha'(1)\ \text{and}\ \alpha'\simeq_{\partial I}c_{\alpha'(0)},\\
n'=0&\Rightarrow \beta'(0)=\beta'(1)\ \text{and}\ \beta'\simeq_{\partial I}c_{\beta'(0)}.
\end{aligned}
\]
For arbitrary
$m,n,m',n'\in\N$, with the zero-length conventions above,
\[
  \alpha\star_{m,n}\beta
  \simeq_{\partial I}
  \alpha'\star_{m',n'}\beta'.
\]
\end{lemma}

\begin{proof}
For $m,n>0$, let $\lambda_{m,n}:I\to I$ be linear on either side of
$m/(m+n)$, fixing $0,1$ and sending that breakpoint to $1/2$. Then
\[
 \alpha\star_{m,n}\beta
   =(\alpha\star_{1,1}\beta)\circ\lambda_{m,n}.
\]
Interpolating $\lambda_{m,n}$ linearly with the identity gives an
endpoint-fixed homotopy to $\alpha\star_{1,1}\beta$.
If $m=0<n$, the hypothesis contracts $\alpha$ to its endpoint constant,
so $\alpha\star_{1,1}\beta$ is homotopic to
$\beta\circ(t\mapsto\max(0,2t-1))$, and then to $\beta$ by interpolating
this parameter with $t$. The case $n=0<m$ uses
$t\mapsto\min(2t,1)$ instead. If both weights vanish, contract both
factors to the common endpoint constant. Thus every weighted expression
is homotopic to its half-and-half concatenation, including all zero cases.
Apply this to both sides, and horizontally concatenate the given
homotopies from $\alpha$ to $\alpha'$ and from $\beta$ to $\beta'$ at
the common breakpoint $1/2$.
\end{proof}

\begin{theorem}[realization soundness]
\label{thm:realization-soundness}
If $p\RwEq q$, then the realized paths are endpoint-fixed homotopic.
\end{theorem}

\begin{proof}
Induct on the derivation of $p\RwEq q$.  Reflexivity is the constant
homotopy, and a named generator is sound by the corresponding presentation
witness.  Symmetry and transitivity use symmetry and vertical composition of
endpoint-fixed homotopies.  In the concatenation-congruence case, apply the
weighted concatenation invariance lemma to the two inductive homotopies and
the four trace lengths; this accounts for rewrites that change word length.
The zero-weight hypotheses of that lemma hold because a trace of length zero
is an empty word, whose realization is a constant path.
For reversal, reverse the path parameter, replacing a homotopy $H(s,t)$ by
$H(s,1-t)$.  Unit, associativity,
double-reversal, and reversal of a composite are literal equalities for the
flat word realization, hence use reflexive homotopies.  The two
inverse-cancellation cases use the standard
endpoint-fixed contraction homotopies made explicit below.  This covers every rule that can occur in a scoped derivation.
\end{proof}

\begin{lemma}[backtracking contracts in every space]
\label{lem:backtracking}
For any continuous path $\alpha:I\to X$, the formula
\[
 H(u,t)=\alpha\!\left(\min\{2t,\,2(1-t),\,1-u\}\right)
 \qquad (u,t\in I)
\]
is an endpoint-fixed homotopy from
$\alpha\star_{1,1}\alpha^{-1}$ to $c_{\alpha(0)}$.
\end{lemma}
\begin{proof}
The minimum is a continuous map $I^2\to I$, so $H$ is continuous.
At $u=0$ its argument is $\min\{2t,2(1-t)\}$, traversing $\alpha$
and then retracing it. At $u=1$ the argument is zero, and at $t=0,1$
it is zero for every $u$. These are the required boundary conditions.
Applying the same construction to $\alpha^{-1}$ gives the other inverse law.
\end{proof}

Because the homotopy stays entirely in $\alpha(I)$, it requires no disk in $X$.
For a nonempty trace $p$, the two factors in $p;p^{-1}$ have equal length,
so its weighted realization is exactly the half-interval concatenation
in Lemma~\ref{lem:backtracking}; the empty case is constant.
Figures~\ref{fig:backtracking} and~\ref{fig:coherence} distinguish these two homotopy mechanisms.
\begin{remark}
The soundness theorem is one-way.  Equality of realized geometric homotopy
classes does not imply scoped rewrite equality unless a separate geometric
completeness theorem is proved.  The presentation can therefore retain computational information.
\end{remark}

\MainFigure{backtracking}{The triangular parameter for an out-and-back path is capped at decreasing heights, ending at the constant path while endpoints stay fixed.}

\section{Scoped quotients and geometric comparison}
\label{sec:quotient}

For coherent traces $u=(p,\gamma)$ and $v=(q,\delta)$, write
$u\approx v$ when their endpoints agree and $p\RwEq q$.  By realization
soundness, equivalent coherent traces determine the same endpoint-fixed
homotopy class of geometric paths.  The quotient arrow space is
\[
  G_{\mathcal P}=T/{\approx},
\]
with the quotient topology induced by $q:T\to G_{\mathcal P}$.  Endpoint
maps descend to continuous maps
\[
  s,t:G_{\mathcal P}\longrightarrow X.
\]
Reflexive paths and reversal descend because the scoped relation is closed
under the corresponding constructors.

The explicit composable carrier is the $T^{(2)}$ defined in Section~\ref{sec:traces}.
Its quotient relation is componentwise scoped equivalence, and its quotient is
denoted $G_{\mathcal P}^{(2),\Fin}$.  Write
\[
  q_{\Fin}^{(2)}:T^{(2)}\longrightarrow G_{\mathcal P}^{(2),\Fin}
\]
for the quotient map.  Concatenation of representatives gives a well-defined
map
\[
  m_{\Fin}:G_{\mathcal P}^{(2),\Fin}\longrightarrow G_{\mathcal P}.
\]
The composite $q\circ\mu:T^{(2)}\to G_{\mathcal P}$ is continuous and
constant on equivalent pairs.  The quotient criterion therefore makes
$m_{\Fin}$ continuous.  This argument uses the quotient of composable
representatives, not a product of quotient maps.

The ordinary composable-pair topology is the subspace topology from
$G_{\mathcal P}\times G_{\mathcal P}$. We call agreement of that
topology with the final pair topology \emph{product-quotient compatibility};
Section~\ref{sec:final-ordinary} gives its equivalent quotient-map conditions.
The final pair topology is part of this construction. We use
\emph{groupoid with final composable-pair topology} for the resulting
structure, reserving \emph{topological groupoid} for continuity of
multiplication on the ordinary pullback. The final topology here is
specifically the quotient of composable representatives.

\begin{theorem}[canonical scoped quotient]
\label{thm:final-groupoid}
For every continuous geometric rewrite presentation $\mathcal P$:
\begin{enumerate}[label=(\alph*),leftmargin=2.2em]
\item $s,t$, the identity map $X\to G_{\mathcal P}$, and reversal
  $G_{\mathcal P}\to G_{\mathcal P}$ are continuous;
\item $m_{\Fin}$ and the two projections
  $G_{\mathcal P}^{(2),\Fin}\to G_{\mathcal P}$ are continuous;
\item the descended operations satisfy the left and right unit laws, inverse
  laws, and associativity.
\end{enumerate}
Thus $G_{\mathcal P}$ is a groupoid with its canonical final
composable-pair topology.
\end{theorem}

\begin{proof}
Lemma~\ref{lem:coherent-operations-continuous} makes the raw operations
continuous. They respect scoped equivalence, so the quotient criterion
gives continuity of endpoints, identity, reversal, and $m_{\Fin}$.
The projections are continuous because
$\mathrm{pr}_i\circ q_{\Fin}^{(2)}=q\circ\mathrm{pr}_i$.

Choose representatives for the algebraic laws. Flat-word units,
associativity, double reversal, and reversal of a composite are literal
trace and realization equalities. The inverse laws instead use the scoped
constructors $p^{-1};p\RwEq 1_b$ and $p;p^{-1}\RwEq 1_a$, sound by the
endpoint-fixed backtracking contractions. Quotient soundness turns these
witnesses into arrow equalities. Descent through the componentwise quotient
makes the operations independent of representatives.
\end{proof}

A semantic interpretation must respect the rewrites and groupoid
operations and forget the choice of coherent representative, as the scoped
quotient does.

\begin{theorem}[semantic factorization]\label{thm:semantic-factorization}
Let $\mathcal H\rightrightarrows X$ be a topological groupoid, and let
$\varphi:T\to\mathcal H_1$ be continuous.  Assume that $\varphi$ preserves
endpoints, identities, reversal, and concatenation, and that every named
rewrite of $\mathcal P$ is sent to an equality in $\mathcal H_1$ for every
choice of coherent representatives.  Assume moreover that $\varphi$ is
invariant under changing the coherent representative: whenever
$(p,\gamma),(p,\delta)\in T$, one has
\[
  \varphi(p,\gamma)=\varphi(p,\delta).
\]
Then there is a unique continuous map
\[
  \overline{\varphi}:G_{\mathcal P}\longrightarrow\mathcal H_1
\]
such that $\overline{\varphi}\,q=\varphi$.  The map
$\overline{\varphi}$ preserves all groupoid operations, with composition on
the source interpreted on $G_{\mathcal P}^{(2),\Fin}$.
\end{theorem}

\begin{proof}
Representative invariance reduces the claim to the canonical pairs
$(p,|p|)$. Induct on scoped rewriting. The named rules hold by hypothesis;
equivalence and congruence follow from equality and preservation of the
operations; structural rules hold in $\mathcal H$. Thus $\varphi$ is
constant on scoped classes. Quotient factorization gives the unique
$\overline\varphi$, and the quotient criterion makes it continuous.
Choosing representatives proves its operation-preservation equations,
including composition on the final pair domain.
\end{proof}

\begin{proposition}[eliminating the coherent coordinate]
\label{prop:coherent-elimination}
Fix a realization model. Give its trace carrier $A$ a topology for which
endpoints and realization are continuous, and give
$T=\{(p,\gamma):\gamma\simeq_{\partial I}|p|\}$ the topology induced
by its inclusion in $A\times X^I$. If scoped equivalence depends only
on the trace, forgetting $\gamma$ induces homeomorphisms on the global
scoped quotients, their fixed-endpoint quotients, and their final
composable-pair spaces. The homeomorphisms preserve endpoints and the
pair projections. They require no geometric completeness hypothesis.
\end{proposition}
\begin{proof}
The maps
\[
 r:T\longrightarrow A,\quad r(p,\gamma)=p,
 \qquad
 z:A\longrightarrow T,\quad z(p)=(p,|p|)
\]
are continuous and preserve endpoints and scoped equivalence. They satisfy
$rz=\mathrm{id}$; the pair $zr(p,\gamma)$ is scoped equivalent to
$(p,\gamma)$ because its trace is unchanged. Quotient descent gives
inverse continuous maps. The maps restrict to fixed endpoint fibers.
Their products restrict to the actual composable subspaces, where the
same identities hold componentwise, so quotient descent also proves
the final-pair statement.
\end{proof}

For the trace-sensitive version, $A=\mathsf{Tr}_{\mathcal P}$ has the
coproduct word topology. For the observable version, give the same trace
set the initial topology induced by $(s(p),t(p),\ell(p),|p|)$.
The induced coherent topology is precisely $T_{\mathrm{obs}}$.
Thus the proposition applies separately to both topologies, and to
each realization model. It does not compare different timings or
identify a based quotient with the subspace of a global quotient.
When $X$ is Hausdorff, the endpoint fibers are closed and saturated, so
these based quotients agree with their global fiber subspaces by the
closed-fiber restriction argument in
Remark~\ref{rem:two-obligations}; Theorems~\ref{thm:circle-winding}
and~\ref{thm:torus-winding} give explicit identifications for the
circle and torus.

We retain $\gamma$ to give the direct continuous projection
$\pi(p,\gamma)=\gamma$ to $X^I$: it supports the universal path section
in Section~\ref{sec:universal} and the coherent normal-form pairs in
Section~\ref{sec:circle-torus} without changing the trace realization.

\subsection{The geometric comparison}
\label{sec:comparison}

Let $H_{\mathrm{geo}}$ be the quotient of coherent traces by equality of
endpoints and endpoint-fixed homotopy of chosen geometric representatives, and
write
\[
  h:T\longrightarrow H_{\mathrm{geo}}
\]
for the quotient map.  Define the final composable geometric domain by
\[
  H_{\mathrm{geo}}^{(2),\Fin}
  =T^{(2)}/\!\sim_{\mathrm{geo}}^{(2)},
  \qquad
  (u,v)\sim_{\mathrm{geo}}^{(2)}(u',v')
  \Longleftrightarrow h(u)=h(u')\ \text{and}\ h(v)=h(v').
\]
Weighted concatenation descends continuously to this final domain.
Its zero-weight hypotheses hold because a coherent representative of an
empty trace is null-homotopic. Identity and reversal likewise descend,
giving $H_{\mathrm{geo}}$ its groupoid structure with final
composable-pair topology.

The realization comparison and its composable-domain map are
\[
  R_{\mathcal P}:G_{\mathcal P}\longrightarrow H_{\mathrm{geo}}
\]
and
\[
  R_{\mathcal P}^{(2)}:G_{\mathcal P}^{(2),\Fin}\longrightarrow
  H_{\mathrm{geo}}^{(2),\Fin},
  \qquad ([u],[v])\longmapsto([u]_{\mathrm{geo}},[v]_{\mathrm{geo}}).
\]
The target contains represented geometric classes, not every ambient
path class: its surjectivity is built into the quotient of $T$.
Write $\Pi_1^q(X)=X^I/{\simeq_{\partial I}}$ for the full path-class
space with its compact-open quotient topology and retained endpoints.
The induced map $H_{\mathrm{geo}}\to\Pi_1^q(X)$ is a continuous injection,
not asserted to be an embedding. The universal presentation represents
all paths; the finite circle and torus presentations cover all based-loop
classes at their chosen basepoints. Soundness makes $R_{\mathcal P}$
well-defined; completeness identifies the two relations on $T$, hence
their quotient topologies.

\begin{theorem}[comparison morphism]
The map $R_{\mathcal P}$ is continuous, surjective on the geometric quotient,
and preserves source, target, identity, and reversal.
The induced pair map $R_{\mathcal P}^{(2)}$ is continuous, and the two maps
commute with composition on the final domains.  In particular,
$R_{\mathcal P}$ is a continuous groupoid morphism for the final composable-pair
structures.
\end{theorem}

\begin{proof}
Soundness makes the comparison well-defined; $h=R_{\mathcal P}\circ q$
and the quotient criterion give continuity. Surjectivity follows from the
definition of $H_{\mathrm{geo}}$. Raw identity and reversal give their
preservation equations. Weighted concatenation invariance gives composition
preservation. The continuous raw pair quotient is constant on componentwise
scoped equivalence, so descent through $q_{\Fin}^{(2)}$ gives the continuous
pair comparison.
\end{proof}

\begin{definition}[geometric completeness]
The presentation is geometrically complete if, for all coherent representatives
$u,v$, equality of their geometric homotopy classes implies scoped equivalence:
\[
  R_{\mathcal P}(q(u))=R_{\mathcal P}(q(v))
  \quad\Longrightarrow\quad u\approx v.
\]
\end{definition}

\begin{theorem}[faithfulness and homeomorphism]
\label{thm:faithfulness-homeomorphism}
The comparison $R_{\mathcal P}$ is faithful, equivalently injective on arrows
because its object map is the identity, if and only if $\mathcal P$ is
geometrically complete.  If $\mathcal P$ is geometrically complete, then
\[
  R_{\mathcal P}:G_{\mathcal P}\xrightarrow{\ \cong\ }H_{\mathrm{geo}}
\]
is a homeomorphism on arrow spaces and an isomorphism of abstract groupoids.
\end{theorem}

\begin{proof}
If the comparison is injective, equality of the geometric homotopy classes
of two coherent representatives gives equality of their scoped classes.
This is precisely geometric completeness.  Conversely, completeness turns equality of comparison images
into equality of scoped classes.  The comparison is always surjective and
continuous.  Under completeness, let
$h:T\to H_{\mathrm{geo}}$ be the geometric quotient map.  If
$h(u)=h(v)$, completeness gives $u\approx v$, hence $q(u)=q(v)$.  Therefore
$q$ factors uniquely through a map
\[
  S:H_{\mathrm{geo}}\longrightarrow G_{\mathcal P},
  \qquad S\circ h=q.
\]
Since $h$ is quotient and $q$ is continuous, $S$ is continuous.  The two
composites are identity maps by quotient induction.
\end{proof}

In the trace-sensitive version, $G_{\mathcal P}$ and $H_{\mathrm{geo}}$ both
carry the quotient topologies from $T_{\mathrm{tr}}$.  The homeomorphism always
compares two quotients of the same carrier.

\subsection{Literal words and the realization boundary}
An endpoint-indexed tree is generated by $1_x$, primitive leaves
$\langle e\rangle$, composable pairs $p;q$, and inverses $p^{-1}$.
Flattening sends these to the empty word at $x$, one positive letter,
word concatenation, and reversed signed word, respectively. Length is
the length of the flattened word. Canonical reconstruction associates
a nonempty word to the right, using inverses for negative letters;
it sends the empty word at $x$ to $1_x$. Tree scoped rewriting uses the
same structural constructors, with the named word rules imposed on
their canonical reconstructions.
The binary realization recursively sends $1_x$ to $c_x$, a primitive
leaf to $\rho(e)$, a composite to half-and-half concatenation, and an
inverse to path reversal; this is the Lean carrier using \LeanName{Path.trans}.
The separate flat realization assigns each primitive letter exactly one
$1/n$ slot and omits zero-length factors. Trees with the same signed word
have equal flat parametrized realizations; flat and binary realizations
are endpoint-fixed homotopic. The latter assertion supplies semantic
soundness and does not compare observable quotient topologies.

For either realization model, the tree observable trace topology is induced
by $(s(p),t(p),\ell(p),|p|)$; its full-word refinement also records the
flattened word in $\mathsf{Tr}_{\mathcal P}$. Coherent carriers additionally
record $\gamma$. For flat realization the full-word trace topology is
induced by flattening alone, since realization is continuous and
word-determined. The binary refinement still retains exact binary
realization, which can depend on the parentheses.

\begin{proposition}[literal-word correspondence]
\label{prop:literal-word-correspondence}
For the flat realization, flattening and canonical reconstruction identify
tree scoped rewriting with the independently generated relation on
composable signed words using adjacent inverse cancellation and the
flattened named rules. They induce homeomorphisms of the corresponding
fixed-endpoint and global scoped quotients and final composable-pair
domains. For the full-word topology, these are the genuine tuple-stratum
and coproduct topologies, with $W_0=X$.
\end{proposition}
\begin{proof}
Units, associativity, and reversal normalize each tree to the canonical
reconstruction of its signed word. Adjacent cancellation has a structural witness in the tree relation.
Each flattened named rule uses its named constructor, composed with
structural normalization. Induction on scoped tree rewriting gives the
converse. Equality of geometric paths is never
used as an additional rewrite. Reconstruction is continuous on each
finite tuple stratum. The induced code retains the complete signed word
and the source; the source is redundant on positive strata and records the
basepoint on $W_0$. Flattening and reconstruction descend to continuous
inverse maps. Restriction to the actual endpoint-matched representative
subspace gives the final-pair comparison. The endpoint-preserving arrow homeomorphism also identifies the ordinary
pullbacks. Together with the final-pair map, it transfers product-quotient
compatibility (Corollary~\ref{cor:compatibility-invariance}). It does not
establish that compatibility without another argument, or identify a
based quotient with a global fiber subspace.
\end{proof}

\begin{proposition}[trace-sensitive refinement]
\label{prop:trace-sensitive-refinement}
Let $T_{\mathrm{tr}}$ and $T_{\mathrm{obs}}$ denote the two topologies above.
The identity $i_T:T_{\mathrm{tr}}\to T_{\mathrm{obs}}$ is continuous, and its
restriction to the common composable set is continuous:
\[
  i_T^{(2)}:T_{\mathrm{tr}}^{(2)}\longrightarrow T_{\mathrm{obs}}^{(2)}.
\]
Let $G_{\mathcal P}^{\mathrm{tr}}$ and $G_{\mathcal P}^{\mathrm{obs}}$ be
the quotient arrow spaces for the same scoped relation, formed from the two
representative topologies.  Then the identity on quotient sets induces a
continuous bijection
\[
  J:G_{\mathcal P}^{\mathrm{tr}}\longrightarrow G_{\mathcal P}^{\mathrm{obs}},
  \qquad J\,q_{\mathrm{tr}}=q_{\mathrm{obs}}\,i_T.
\]
The corresponding map on final composable-pair quotients is continuous as
well.  In general $J^{-1}$ need not be continuous, so $J$ need not be a
homeomorphism, as the expanded presentation in
Corollary~\ref{new-cor:intrinsic-abbreviation-change} shows.
\end{proposition}

\begin{proof}
The factorization of $\kappa_{\mathrm{obs}}$ through $\kappa_{\mathrm{tr}}$
proves continuity of $i_T$ by the initial-topology criterion.  Restriction to
the endpoint pullback gives $i_T^{(2)}$.  The identity preserves the scoped
equivalence relation, so it descends to $J$.  Since
$q_{\mathrm{tr}}$ is a quotient map and
$J\,q_{\mathrm{tr}}=q_{\mathrm{obs}}\,i_T$ is continuous, the quotient
criterion proves continuity of $J$.  The same argument applies to the
componentwise final composable-pair quotients.
\end{proof}

\section{Presentation invariance and observable dependence}
\label{sec:functoriality}

\subsection{Continuous presentation maps}

Let $\mathcal P$ and $\mathcal Q$ be presentations on $(X,E)$ and $(Y,F)$.
A presentation map consists of a continuous map $f:X\to Y$, a continuous
step map preserving endpoints and realizations, and a proof that every named
rule of $\mathcal P$ maps to a scoped derivation in $\mathcal Q$.
Presentation maps form a category: the identity has the identity space and step
maps, and composition is the composite of the corresponding space and step
maps, with the scoped-derivation proofs composed by induction.

\begin{theorem}[functoriality]
Every presentation map $F:\mathcal P\to\mathcal Q$ induces a continuous map
\[
  F_1:G_{\mathcal P}\longrightarrow G_{\mathcal Q}
\]
that preserves source, target, identities, reversal, and final-domain
composition.  The assignment preserves identity maps and composites.
\end{theorem}

\begin{proof}
Write $f:X\to Y$ for the underlying map of spaces and $F_*p$ for the trace
obtained by applying the step map to every letter of $p$.  The raw map is
\[
  \widetilde F:T_{\mathcal P}\longrightarrow T_{\mathcal Q},
  \qquad (p,\gamma)\longmapsto(F_*p,f\circ\gamma).
\]
It preserves coherence because realizations are preserved by the step map and
endpoint-fixed homotopies may be postcomposed with $f$.  It is continuous:
on every word stratum the induced trace map is continuous, postcomposition
$X^I\to Y^I$ is continuous for the compact-open topology, and the observable
code of $\widetilde F(p,\gamma)$ is consequently a continuous function of the
observable code of $(p,\gamma)$.  Induction on scoped rewrite derivations
proves that $\widetilde F$ sends related traces to related traces.  It therefore
descends through the quotient, and the quotient criterion gives continuity on
arrows.  Applying the step map to reflexivity, reversal, and concatenation
gives the corresponding preservation equations before quotienting.
The same argument on the explicit composable carrier gives a continuous map on
final composable domains and proves final composition preservation.  Identity
and composite step maps are definitionally the identity and composite maps on
raw traces; quotient extensionality completes the proof.
\end{proof}

The same construction gives ordinary-pair continuity for every presentation
map.  If both source and target presentations satisfy product-quotient
compatibility, the ordinary multiplication maps are continuous as well.  Continuity of the induced map on arrows does not require this
additional compatibility assumption.  For the trace-sensitive topology, the
word component of $\widetilde F$ is continuous on each coproduct stratum and
the path component is the same postcomposition map, so the proof carries
over.

\subsection{Continuous finite-word interpretations}
\label{sec:word-interpretations}

A primitive may instead be interpreted as a finite trace with the same
endpoints. Extend such a substitution $M$ through identities, composition,
and reversal. If images of every named source rule are target scoped
rewrites, induction proves preservation of the entire scoped relation.
Primitive roundtrip rewrites for $M,N$ then extend to every trace. This is
the algebraic part of a presentation change. It does not by itself establish
continuity when length and exact realization are observable.

\begin{theorem}[presentation invariance from primitive word codes]
\label{thm:flat-word-invariance}
Let $M,N$ be endpoint-preserving finite-trace substitutions between
presentations over the same space $X$. Suppose they preserve named rules and
their primitive roundtrips are scoped rewrites to the original primitives.
Suppose also that the two maps sending a primitive to the full signed word
of its substituted trace are continuous into the coproduct of finite word
spaces. Then they induce homeomorphisms of the flat trace-sensitive scoped
quotients at fixed endpoints, of global arrow quotients, and of the final
composable-pair domains. The pair maps commute with the two arrow projections.
For discrete primitive alphabets the continuity hypotheses hold automatically.
\end{theorem}
\begin{proof}
On each finite input stratum, substitution concatenates finitely many output
words. Continuity of the primitive word codes makes this operation continuous;
the discrete length coordinate isolates each output-length stratum. Equal-slot
execution is continuous on the compatible finite-product strata, including
the empty stratum $X$. Thus the output full word, endpoints, and realized path
are continuous functions of the input source and full word. Choosing that
realized path as the coherent representative gives a continuous map of the
global carriers. Restricting its product to matching endpoints gives the
continuous map of actual composable representative subspaces.
Induction extends the primitive roundtrip rewrites to arbitrary traces.
Quotient descent therefore gives inverse continuous maps in all three cases.
The pair projection equations already hold on representatives.
\end{proof}

\begin{corollary}[invariance of product-quotient compatibility]
\label{cor:compatibility-invariance}
Under the hypotheses of Theorem~\ref{thm:flat-word-invariance}, the
two trace-sensitive presentations have product-quotient compatibility
simultaneously. The same conclusion holds for the literal-word
correspondence of Proposition~\ref{prop:literal-word-correspondence}.
\end{corollary}
\begin{proof}
An endpoint-preserving arrow homeomorphism restricts on its square to
a homeomorphism of ordinary composable-pair spaces. The theorems also
supply a homeomorphism of the final pair spaces. Since both are induced
by the same substitutions, they commute with the canonical final-to-ordinary
comparison maps. One comparison is a homeomorphism exactly when the
other is. Theorem~\ref{thm:ordinary-final} identifies this condition
with product-quotient compatibility.
\end{proof}

This transfers compatibility between the two presentations; an additional
argument is still needed to establish it for either one. Compatibility
with the geometric comparison also needs a separate hypothesis: each
substituted primitive realization must be homotopic to the original path.
Induction and Lemma~\ref{lem:weighted-concat-invariance} then extend
this condition to every trace.

For a nondiscrete application, take $X=\mathbb R\times\Circle$ and the
primitive family $E=\mathbb R$, with $e_x$ based at $(x,0)$ and realizing
$t\mapsto(x,[t])$. Declare no named rules. Enlarge the alphabet to
$E'=\mathbb R\sqcup\mathbb R$ by adding $a_x$, realizing
$t\mapsto(x,[2t])$, with the defining rule $a_x\RwEq e_x;e_x$.
The realization families are continuous, and the defining rule is sound by
exact equal-slot realization. Inclusion and elimination $a_x\mapsto e_x;e_x$
preserve the rules; their primitive roundtrips are identities or instances
of the defining rule. The elimination word code is the continuous map
$x\mapsto(e_x)$ on one component and the diagonal map $x\mapsto(e_x,e_x)$
on the other. Inclusion is continuous into the one-letter stratum.
Theorem~\ref{thm:flat-word-invariance} therefore gives the global and
final-pair homeomorphisms for these nondiscrete alphabets. Output length is
one on the old component and two on the added component. In general,
continuous word codes force output length to be locally constant, hence
constant on a connected primitive alphabet.

\subsection{Observable quotients over a point}
\label{new-sec:point-observations}

The loss of information in the observable topology can be computed exactly
when all geometric paths are constant. This also distinguishes a change of
quotient topology from a failure of ordinary multiplication.

\begin{theorem}[point presentations and parity]
\label{new-thm:point-parity}
Let $X=\{*\}$ and let $E$ be a nonempty space of primitive labels, all
realizing the unique path in $X$. Allow any set of named word relations,
and write $G$ for the resulting scoped group. Then the observable quotient
has one of the following two topologies.
\begin{enumerate}[label=(\roman*),leftmargin=2.2em]
\item If every named relation preserves word-length parity, the map
\[
  \epsilon:G\longrightarrow\mathbb Z/2,
  \qquad \epsilon([w])=\ell(w)\bmod 2,
\]
is a well-defined surjective homomorphism. The observable topology is the
initial topology induced by $\epsilon$, where $\mathbb Z/2$ is discrete.
Its open sets are exactly
\[
  \varnothing,\qquad \epsilon^{-1}(0),\qquad
  \epsilon^{-1}(1),\qquad G.
\]
\item If some named relation changes word-length parity, the observable
quotient is indiscrete.
\end{enumerate}
In both cases, for every positive integer $k$, the product of the observable
quotient maps
\[
  q^k:T_{\mathrm{obs}}^k\longrightarrow
       (G^{\mathrm{obs}})^k
\]
is quotient. In particular, the final and ordinary composable-pair
topologies agree, and $G^{\mathrm{obs}}$ is a topological group for its
ordinary product topology. If $E$ is discrete, the trace-sensitive quotient
$G^{\mathrm{tr}}$ is discrete.
\end{theorem}

\begin{proof}
There is only one geometric path, so a coherent representative is determined
by its word. All observable coordinates except length are constant.
Consequently an observable-open set of representatives is exactly a union
of complete length strata.

Fix a primitive $e\in E$. If a word represents $g\in G$ and has length
$m$, inserting copies of $e;e^{-1}$ gives representatives of the same
element of every length $m+2j$, $j\geq0$. It follows that if two group
elements have representatives of lengths of the same parity, they have
representatives of one common length. Every observable-open subset of $G$
that contains either element must then contain the other: its inverse image
is a union of complete length strata.

Suppose first that the named relations preserve parity. All structural
rules also preserve parity, and parity is additive under concatenation.
It therefore descends to the displayed homomorphism $\epsilon$.
The empty word and the one-letter word $e$ show that $\epsilon$ is
surjective. The preceding padding argument shows that every open subset
of $G$ is a union of its two fibers. Conversely, the inverse images in
the representative space of these fibers are the unions of the even and
odd length strata, respectively, so both fibers are open in the quotient.
This proves (i).

Now let $u\simeq_{\mathcal P}v$ be a named relation with lengths of
different parity. The word $u^{-1};v$ has odd length and represents the
identity. Appending it shows that every element of $G$ has representatives
of both parities. Any two elements consequently have representatives of a
common length, by padding as above. A nonempty open subset of $G$ must
contain every element, proving (ii).

For the product assertion, an open subset of $T_{\mathrm{obs}}^k$ is a
union of complete strata indexed by length tuples in $\mathbb N^k$.
Let $U\subseteq G^k$ have open inverse image under $q^k$. In case (i),
take two tuples with the same parity in every coordinate. Coordinatewise
padding gives representatives with one common length tuple. Thus membership
in $U$ is constant on each fiber of $\epsilon^k$. These fibers are open
in the ordinary product topology, so $U$ is open. In case (ii), every
coordinate has representatives of both parities, and the same argument
shows that $U$ is empty or all of $G^k$, hence is open. The reverse
implication follows from continuity of $q^k$, so $q^k$ is quotient.
For $k=2$ all pairs are composable, giving product-quotient compatibility;
continuity of the group operations follows from the quotient construction.

Finally, if $E$ is discrete, every finite word stratum is discrete, as is
their coproduct. The coherent path coordinate is a singleton. The
trace-sensitive representative space and its quotient are therefore
discrete.
\end{proof}

\begin{example}[cyclic presentations]
\label{new-ex:cyclic-parity}
At a point, take one discrete primitive $e$ and impose $e^m\simeq1$ for
$m\geq1$. The scoped group is cyclic of order $m$. For even $m$ its
observable topology has the two parity classes as its minimal nonempty
open sets; for odd $m$ it is indiscrete. The trace-sensitive topology is
discrete in either case. Thus finiteness of the algebraic group does not
force the observable topology to be discrete.
\end{example}

\begin{corollary}[a conservative abbreviation changes the quotient space]
\label{new-cor:intrinsic-abbreviation-change}
An algebraically conservative generator abbreviation can turn a discrete
two-element equal-slot observable quotient into an indiscrete two-element
quotient. In particular, it can change the homeomorphism type and destroy
the $T_0$ and Hausdorff properties. Both presentations can nevertheless satisfy
product-quotient compatibility, and their trace-sensitive quotients can
remain homeomorphic.
\end{corollary}

\begin{proof}
Work over a point with one discrete primitive $e$ and the named relation
$e;e\RwEq1$. The scoped group is $C_2$. All named relations preserve
parity, so Theorem~\ref{new-thm:point-parity} gives the discrete topology:
the parity homomorphism $C_2\to\mathbb Z/2$ is a bijection.
Adjoin a discrete primitive $a$ and the defining relation $a\RwEq e;e$.
Inclusion and elimination $a\mapsto e;e$ induce inverse algebraic
quotient maps, since their primitive roundtrips are identities or the
defining relation. The expanded group is still $C_2$.
The new relation changes length parity, so the observable topology is
indiscrete. A discrete and an indiscrete two-element space cannot be
homeomorphic. The theorem gives ordinary product compatibility in both
presentations and discreteness of both trace-sensitive quotients.
Thus their algebraic equivalence remains a trace-sensitive homeomorphism.
\end{proof}

Figure~\ref{fig:redundant-generator} displays this change. Omitting
the original relation $e;e\RwEq1$ gives the same phenomenon on
$\mathbb Z$: the old observable topology has four open sets, while
the expanded one is indiscrete. It also separates two
issues: loss of observable information under a presentation change can
occur even when multiplication is continuous on the ordinary product
throughout. The open-arrow criterion is sufficient but not necessary here.
In the expanded two-element example, the zero-length stratum is open
and its image is the identity singleton, which is not observable-open;
nevertheless every finite product of the quotient map is quotient.

\begin{figure}[htbp]
\centering
\begin{tikzpicture}[font=\small,>=Stealth]
\node[font=\bfseries] at (0,2) {Before abbreviation};
\node[font=\bfseries] at (6,2) {After abbreviation};
\node at (0,1.45) {$\langle e\mid e^2=1\rangle$};
\node at (6,1.45) {$\langle e,a\mid e^2=1,\ a=e^2\rangle$};
\draw[rounded corners,draw=pathblue,thick] (-1.4,-.15) rectangle (-.3,.8);
\draw[rounded corners,draw=pathblue,thick] (.3,-.15) rectangle (1.4,.8);
\fill (-.85,.4) circle (1.6pt);
\fill (.85,.4) circle (1.6pt);
\node[below] at (-.85,.33) {$[1]$};
\node[below] at (.85,.33) {$[e]$};
\draw[rounded corners,draw=pathorange,thick] (4.5,-.15) rectangle (7.5,.8);
\fill (5.2,.4) circle (1.6pt);
\fill (6.8,.4) circle (1.6pt);
\node[below] at (5.2,.33) {$[1]=[a]$};
\node[below] at (6.8,.33) {$[e]$};
\draw[->,thick] (1.75,.55) -- (4.1,.55)
  node[midway,above] {same group $C_2$};
\node[align=center] at (0,-.9) {Discrete topology\\four open sets};
\node[align=center] at (6,-.9) {Indiscrete topology\\two open sets};
\node[align=center] at (3,-1.8)
  {Both have continuous ordinary multiplication.\\
   Both trace-sensitive quotients are discrete.};
\end{tikzpicture}
\caption{A conservative abbreviation changes the homeomorphism type
of the observable quotient (Corollary~\ref{new-cor:intrinsic-abbreviation-change}).
Every path is constant. Before adjoining $a=e^2$, length parity separates
the two group elements. Afterward, the one-letter representatives $a,e$
have identical observations but represent different elements.
Boxes show the minimal nonempty open sets in each two-element quotient.}
\FigureAltText{The same two-element group has two singleton open neighborhoods
before abbreviation and only the whole space as a nonempty open set afterward.
The defining relation a equals e squared preserves the group but destroys
observable separation.}
\label{fig:redundant-generator}
\end{figure}

\FloatBarrier

\subsection{Parametrization dependence}
Figure~\ref{fig:timing-models} compares the two execution schedules.
\begin{proposition}[binary timing can identify observable codes]
\label{prop:binary-timing-collision}
There is a three-generator based circle presentation with no named rewrite
rules for which the observable quotients of the fixed-base carriers differ:
the binary quotient is not $T_0$, and the equal-slot flat quotient is discrete.
\end{proposition}
\begin{proof}
Let $\alpha(t)=[t]$ in $\Circle$, and let the primitive labels $a,b,c$
realize $\alpha$, $\beta=\alpha\star\alpha$, and
$\chi=\alpha\star\beta$, respectively, with $\star$ denoting ordinary
half-and-half concatenation. The binary traces
$p=(a;b);b$ and $q=c;(a;a)$ both realize $\chi\star\beta$.
They have length three and can choose this same path as coherent
representative, so their binary observable codes agree. Signed occurrence
count of $c$ is invariant under structural scoped rewrites and has values
zero and one on these two traces. Their scoped classes are distinct.
Every open set in the induced raw observable topology contains both chosen
representatives or neither; the same holds for their classes by the
definition of quotient topology. Hence that quotient is not $T_0$.

For the flat construction, the six oriented primitive paths have distinct
winding numbers $1,2,3,-1,-2,-3$. Restriction to each prescribed equal slot
therefore recovers the signed letter. At each fixed length there are finitely
many words, with distinct realizations in the Hausdorff compact-open path
space. This finite image is discrete, so the fiber of each observed word is
open. The discrete length coordinate separates the strata. A scoped-class
preimage is a union of word fibers, even when the coherent path varies;
thus the fixed-base flat quotient is discrete. This last assertion concerns
the quotient of the based carrier, not an unproved subspace comparison with
the global quotient.
\end{proof}

\begin{figure}[htbp]
\centering
\begin{tikzpicture}[font=\small,>=Stealth]
\node[anchor=east] at (-0.25,1.8) {Binary $(a;b);c$};
\node[anchor=east] at (-0.25,0.65) {Binary $a;(b;c)$};
\node[anchor=east] at (-0.25,-0.5) {Flat $a;b;c$};
\foreach \left/\right/\letter/\shade in {0/1.5/a/pathblue,1.5/3/b/pathteal,3/6/c/pathorange}{
  \draw[fill=\shade!12,draw=\shade] (\left,1.55) rectangle (\right,2.05);
  \node at ({(\left+\right)/2},1.8) {$\letter$};}
\foreach \left/\right/\letter/\shade in {0/3/a/pathblue,3/4.5/b/pathteal,4.5/6/c/pathorange}{
  \draw[fill=\shade!12,draw=\shade] (\left,0.4) rectangle (\right,0.9);
  \node at ({(\left+\right)/2},0.65) {$\letter$};}
\foreach \left/\right/\letter/\shade in {0/2/a/pathblue,2/4/b/pathteal,4/6/c/pathorange}{
  \draw[fill=\shade!12,draw=\shade] (\left,-0.75) rectangle (\right,-0.25);
  \node at ({(\left+\right)/2},-0.5) {$\letter$};}
\foreach \y in {1.35,0.2}{
  \node at (0,\y) {$0$}; \node at (1.5,\y) {$\frac14$};
  \node at (3,\y) {$\frac12$}; \node at (4.5,\y) {$\frac34$};
  \node at (6,\y) {$1$};}
\node at (0,-0.95) {$0$}; \node at (2,-0.95) {$\frac13$};
\node at (4,-0.95) {$\frac23$}; \node at (6,-0.95) {$1$};
\end{tikzpicture}
\caption{Exact timing depends on the trace model. Recursive binary
concatenation halves each concatenation interval;
equal-slot realization assigns one third to each letter of this word,
independently of parentheses. Homotopy between these paths does not identify
topologies that observe their exact parametrizations.}
\FigureAltText{Two parenthesizations of three letters have binary durations
one quarter, one quarter, one half and one half, one quarter, one quarter.
The flat model gives every letter one third.}
\label{fig:timing-models}
\end{figure}

\subsection{Completeness and the two quotient topologies}

\begin{proposition}[when trace sensitivity disappears]
\label{prop:trace-collapse-section}
Let $P_{\mathcal P}\subseteq X^I$ be the space of paths represented by
coherent pairs, with its subspace topology.  Write
$\pi_{\mathrm{obs}}:T_{\mathrm{obs}}\to P_{\mathcal P}$ and
$\pi_{\mathrm{tr}}:T_{\mathrm{tr}}\to P_{\mathcal P}$ for the maps that
forget the trace.  Suppose that $\mathcal P$ is geometrically complete and
that $\pi_{\mathrm{tr}}$ has a continuous section
$s:P_{\mathcal P}\to T_{\mathrm{tr}}$.  Then the quotient comparison
$J:G_{\mathcal P}^{\mathrm{tr}}\to G_{\mathcal P}^{\mathrm{obs}}$ is a
homeomorphism.
\end{proposition}

\begin{proof}
The observable code makes $\pi_{\mathrm{obs}}$ continuous.  Thus
$\theta=s\pi_{\mathrm{obs}}:T_{\mathrm{obs}}\to T_{\mathrm{tr}}$ is
continuous.  For each coherent pair $u$, the chosen paths of $u$ and
$\theta(u)$ are equal.  Geometric completeness makes their scoped classes
equal, so $q_{\mathrm{tr}}\theta=J^{-1}q_{\mathrm{obs}}$.  Since
$q_{\mathrm{obs}}$ is quotient, $J^{-1}$ is continuous.  The forward map
is continuous by Proposition~\ref{prop:trace-sensitive-refinement}.
\end{proof}

The hypothesis asks for a continuous choice of trace for each represented
path. The universal one-letter section supplies one. Completeness alone
does not, as the next example shows.

\begin{proposition}[completeness without trace-topology collapse]
\label{prop:complete-trace-separation}
There is a geometrically complete presentation for which the comparison
$J:G_{\mathcal P}^{\mathrm{tr}}\to G_{\mathcal P}^{\mathrm{obs}}$ is not a
homeomorphism.  On the quotient formed from its fixed-endpoint coherent
carrier at one base point, the trace-sensitive topology is discrete and the
observable topology is indiscrete.
\end{proposition}

\begin{proof}
Let $X$ be the harmonic archipelago, with base point $a$.  Take $E$ to be the
set $X^I$ with the discrete topology, and let each step realize the path it
labels.  The endpoint and realization maps are continuous because $E$ is
discrete.  Declare a rewrite between every pair of parallel traces whose
realizations are endpoint-fixed homotopic.  Soundness and the declared rules
make this presentation geometrically complete.

In the fixed-endpoint trace carrier, each trace is an isolated point: this
holds for positive-length words because $E$ is discrete, and for the empty
trace because its endpoint is fixed at $a$.  For a scoped class $C$, its
inverse image in $T_{\mathrm{tr}}(a,a)$ is the union of the full coherent
fibers over traces realizing $C$.  Each such fiber is open in the
trace-sensitive carrier.  Thus every singleton is open in the based
trace-sensitive quotient.

The projection $(p,\gamma)\mapsto\gamma$ from the observable based carrier
to $X^I(a,a)$ is continuous.  It has the continuous section
$\gamma\mapsto(\langle\gamma\rangle,\gamma)$, whose observable code is
$(a,a,1,\gamma,\gamma)$.  It is therefore a quotient map.  Completeness
identifies the scoped relation with endpoint-fixed homotopy, so the based
observable quotient is homeomorphic to $\pi_1^q(X,a)$.  Fabel proved that
this group is uncountable and indiscrete
\cite[Corollaries~9--10]{Fabel05Archipelago}; the
uniform loop topology used there agrees with the compact-open topology for
his metric model.  The based comparison is a continuous bijection from a discrete
space to a non-singleton indiscrete space.

The global comparison fails as well.  Choose a nonidentity class $C$ in
$\pi_1^q(X,a)$.  Its inverse image under the global trace-sensitive quotient
contains no empty traces and is a union of full coherent fibers over
positive-length words, hence is open.  If $\{C\}$ were open in the global
observable quotient, its pullback along the continuous one-letter section
$X^I(a,a)\to T_{\mathrm{obs}}$ would make the homotopy class $C$ open in
$X^I(a,a)$, contrary to the indiscrete quotient topology.
\end{proof}

\section{Final and ordinary composable-pair topologies}
\label{sec:final-ordinary}

As before, $G_{\mathcal P}=G_{\mathcal P}^{\mathrm{obs}}$ and
$T=T_{\mathrm{obs}}$.  The definitions and theorems of this section apply
equally to the trace-sensitive quotient.  Compatibility for one
representative topology does not imply compatibility for the other, however;
the continuous refinement map $J$ alone does not establish a transfer
of compatibility.

There is another natural composable domain.  Let
\[
  G_{\mathcal P}^{(2),\Ord}
  =\{(x,y)\in G_{\mathcal P}\times G_{\mathcal P}:t(x)=s(y)\}
\]
with the subspace topology.  The underlying set of this space is the same as
the underlying set of the final domain, but its topology need not be the same.
The canonical ordinary-pair map is
\[
  q_{\Ord}^{(2)}:T^{(2)}\longrightarrow G_{\mathcal P}^{(2),\Ord}.
\]
It is continuous and surjective.  There is a unique bijection
\[
  j:G_{\mathcal P}^{(2),\Fin}\longrightarrow G_{\mathcal P}^{(2),\Ord}
\]
such that
\[
  q_{\Ord}^{(2)}=j\circ q_{\Fin}^{(2)}.
\]
The map $j$ is continuous because $q_{\Fin}^{(2)}$ is quotient.  Products of
quotient maps need not be quotient maps, so $q_{\Ord}^{(2)}$ being quotient is
an extra property; it does not follow formally from $q$ being quotient.

\begin{definition}[product-quotient compatibility]
The presentation has product-quotient compatibility if $q_{\Ord}^{(2)}$ is a
quotient map onto the ordinary composable-pair space.  Since
$q_{\Ord}^{(2)}=j\circ q_{\Fin}^{(2)}$ and $q_{\Fin}^{(2)}$ is quotient, this
is equivalent to $j$ being a quotient map.
\end{definition}

The domains have the same composable classes. Final openness is tested by
pullback to $T^{(2)}$; ordinary openness comes from $G_{\mathcal P}\times G_{\mathcal P}$.
Compatibility identifies these topologies (Figure~\ref{fig:final-ordinary}).

\begin{theorem}[ordinary/final comparison]
\label{thm:ordinary-final}
For every presentation $\mathcal P$, the following are equivalent:
\begin{enumerate}[label=(\arabic*),leftmargin=2.2em]
\item $q_{\Ord}^{(2)}:T^{(2)}\to G_{\mathcal P}^{(2),\Ord}$ is a quotient map;
\item $j:G_{\mathcal P}^{(2),\Fin}\to G_{\mathcal P}^{(2),\Ord}$ is a
  quotient map;
\item the inverse bijection $j^{-1}:G_{\mathcal P}^{(2),\Ord}\to
  G_{\mathcal P}^{(2),\Fin}$ is continuous;
\item $j$ is a homeomorphism, equivalently the final and ordinary
  composable-pair topologies coincide.
\end{enumerate}
Thus each condition is also equivalent to product-quotient compatibility.
When these conditions hold, multiplication is continuous from the ordinary
pullback domain.  Theorem~\ref{thm:final-groupoid} uses none of these
conditions.
\end{theorem}

\begin{proof}
If $j$ is quotient, so is the composite $q_{\Ord}^{(2)}=j\circ q_{\Fin}^{(2)}$.
Conversely, suppose $q_{\Ord}^{(2)}$ is quotient and $j^{-1}(U)$ is open.  Then
$(q_{\Ord}^{(2)})^{-1}(U)=(q_{\Fin}^{(2)})^{-1}(j^{-1}(U))$ is open, so $U$ is
open.  This proves that (1) and (2) are equivalent.

The map $j$ is a continuous bijection.  A bijective quotient map has
continuous inverse, so (2) and (3) are equivalent; together with the already
known continuity of $j$, these are equivalent to (4).  The homeomorphism
condition is precisely equality of the two topologies on the common
underlying set.  Finally, when these conditions hold, multiplication is
continuous by the quotient criterion applied to its continuous
representative-level concatenation.
\end{proof}

\begin{remark}[comparison with convenient categories]
One can instead change the ambient category of spaces, for example by using
compactly generated spaces and their $k$-ified products \cite{May99}.  That
remedy changes the topology assigned to the ordinary product.  The present
construction stays in $\Top$ and records the final topology forced by the
chosen representative quotient, so the failure of the ordinary product
remains visible.  We do not claim that the two remedies agree in general.
\end{remark}

\begin{figure}[htbp]
\centering
\begin{tikzcd}[column sep=3.6em,row sep=3em]
 T^{(2)} \arrow[r,"q_{\Fin}^{(2)}"] \arrow[d,"q_{\Ord}^{(2)}"'] &
 G_{\mathcal P}^{(2),\Fin} \arrow[d,"j"] \arrow[r,"m_{\Fin}"] &
 G_{\mathcal P} \arrow[d,equal] \\
 G_{\mathcal P}^{(2),\Ord} \arrow[r,equal] &
 G_{\mathcal P}^{(2),\Ord} \arrow[r,"m_{\Ord}"'] &
 G_{\mathcal P}
\end{tikzcd}
\caption{The composition-continuity boundary.
The upper multiplication is always continuous.  The comparison $j$ is
always a continuous bijection, and is a homeomorphism exactly under
Theorem~\ref{thm:ordinary-final}'s equivalent quotient conditions.
These conditions imply continuity of $m_{\Ord}$.  Continuity of $m_{\Ord}$ alone is not asserted to imply these
quotient conditions.}
\FigureAltText{The same composable arrows carry a final and an ordinary topology. Their continuous bijection is a homeomorphism exactly under product-quotient compatibility.}
\label{fig:final-ordinary}
\end{figure}
\FloatBarrier

\begin{theorem}[open arrow quotients]
\label{thm:open-arrow-quotient}
If the arrow quotient $q:T\to G_{\mathcal P}$ is open, then
$q_{\Ord}^{(2)}:T^{(2)}\to G_{\mathcal P}^{(2),\Ord}$ is an open quotient
map.  Consequently the final and ordinary composable-pair topologies agree,
and multiplication is continuous on the ordinary pullback.
\end{theorem}

\begin{proof}
The product $q\times q:T\times T\to G_{\mathcal P}\times G_{\mathcal P}$
is an open surjection.  Because scoped equivalence preserves endpoints,
$T^{(2)}$ is the full inverse image of
$G_{\mathcal P}^{(2),\Ord}$ under this product map.  Restricting an open map
to the full inverse image of a subspace gives an open map onto that subspace:
the image of $U\cap T^{(2)}$, for $U$ open in $T\times T$, is
$(q\times q)(U)\cap G_{\mathcal P}^{(2),\Ord}$.  This restriction is
$q_{\Ord}^{(2)}$, and it is continuous and surjective.  It is therefore
quotient, so Theorem~\ref{thm:ordinary-final} applies.
\end{proof}

\begin{theorem}[compact-Hausdorff compatibility]
Suppose that the final composable domain
$G_{\mathcal P}^{(2),\Fin}$ is compact and that the ordinary composable domain
$G_{\mathcal P}^{(2),\Ord}$ is Hausdorff.  Then product-quotient compatibility
holds.  Hence the final and ordinary composable-pair topologies agree and
ordinary-pullback multiplication is continuous.
\end{theorem}

\begin{proof}
The comparison map
\[
  j:G_{\mathcal P}^{(2),\Fin}\longrightarrow
  G_{\mathcal P}^{(2),\Ord}
\]
is a continuous bijection.  A continuous bijection
from a compact space to a Hausdorff space is a homeomorphism: its image of
every closed set is compact, hence closed, so its inverse is continuous.  The
composite $j\circ q_{\Fin}^{(2)}=q_{\Ord}^{(2)}$ is therefore a quotient map.
This is exactly product-quotient compatibility.
\end{proof}

\begin{corollary}[discrete arrow spaces]
\label{cor:discrete}
If $G_{\mathcal P}$ is discrete, then the final and ordinary composable-pair
topologies coincide, and both are discrete.
\end{corollary}

\begin{proof}
The ordinary domain is a subspace of the discrete space
$G_{\mathcal P}\times G_{\mathcal P}$, so it is discrete.  Every subset of the
final domain is then the preimage under the continuous map $j$ of an open
set, hence open.  So both topologies are discrete and $j$ is a
homeomorphism.
\end{proof}

In both results, Theorem~\ref{thm:ordinary-final} turns agreement of the two
topologies into continuity of ordinary multiplication.
Theorem~\ref{thm:final-groupoid} needs neither hypothesis.

\section{The universal presentation}
\label{sec:universal}

For every topological space $X$, take the primitive step space to be the
continuous interval path space $X^I$.  Its endpoint maps are evaluation at
$0$ and $1$.  Declare two parallel traces to be named generators exactly when
their realized paths are endpoint-fixed homotopic.
Every continuous path is available as a one-letter primitive trace, and
completeness is built into the declared rules rather than derived from a
normalizer.

The two lemmas below treat both representative topologies explicitly.
Without a superscript, $T_{\mathcal U_X}$ and $G_{\mathcal U_X}$ refer to the
observable choice.

\begin{lemma}[coherent-path quotient identification]
\label{lem:coherent-path-quotient-identification}
For the universal presentation, the projection
\[
  \pi:T_{\mathcal U_X}\longrightarrow X^I,
  \qquad \pi(p,\gamma)=\gamma,
\]
is a continuous quotient map.  Consequently, the quotient of
$T_{\mathcal U_X}$ by endpoint-fixed homotopy of the chosen representative is
canonically homeomorphic to the usual quotient-topologized path-class space
\[
  \Pi_1^q(X)=X^I/{\simeq_{\partial I}},
\]
with the endpoint maps retained.
\end{lemma}

\begin{proof}
The map $\pi$ is a continuous coordinate projection for either topology.
Its section $\sigma(\gamma)=(\langle\gamma\rangle,\gamma)$ is coherent
by the constant homotopy and is continuous: its trace coordinate lies in
the one-letter component $X^I$, and its observable code is
$(\gamma(0),\gamma(1),1,\gamma,\gamma)$.
Since $\pi\sigma=\mathrm{id}$, openness of $\pi^{-1}(U)$ implies
openness of $U=\sigma^{-1}(\pi^{-1}(U))$. Thus $\pi$ is quotient.
Both maps preserve endpoint-fixed homotopy of the chosen paths and
induce continuous maps of the two quotients. Their composites are
identities there: $\pi\sigma$ already is an identity, while
$(p,\gamma)$ and $\sigma\pi(p,\gamma)$ have the same chosen path.
This proves the claimed homeomorphism, preserving endpoints.
\end{proof}

\begin{lemma}[fiberwise coherent-path quotient]
\label{lem:fiberwise-coherent-path-quotient}
Fix $a,b\in X$, and let
\[
  X^I(a,b)=\{\gamma\in X^I:\gamma(0)=a,\ \gamma(1)=b\}
\]
carry its subspace topology.  For $*=\mathrm{tr},\mathrm{obs}$, define the
coherent carrier with these fixed endpoints
\[
  T_{\mathcal U_X}^{*}(a,b)=
  \{(p,\gamma)\in T_{\mathcal U_X}^{*}:s(p)=a,\ t(p)=b\}
\]
with its own subspace topology, and define
$G_{\mathcal U_X}^{*}(a,b)$ to be its quotient by scoped equivalence, with
the quotient topology.  The restricted projection
\[
  \pi_{a,b}^{*}:T_{\mathcal U_X}^{*}(a,b)\longrightarrow X^I(a,b)
\]
is quotient.  Hence $G_{\mathcal U_X}^{*}(a,b)$ is canonically homeomorphic
to the endpoint-fixed homotopy quotient of $X^I(a,b)$.
\end{lemma}

\begin{proof}
The continuous maps $\pi$ and $\sigma$ restrict to the displayed endpoint
fibers and still satisfy $\pi_{a,b}^{*}\sigma_{a,b}=\mathrm{id}$.
Consequently, the same open-set argument proves directly that
$\pi_{a,b}^{*}$ is quotient. Coherence and the universal named rules
identify homotopy of the chosen paths with scoped equivalence on this
fiber. The restricted maps therefore induce inverse continuous quotient
maps as above. No restriction theorem for a global quotient map is used.
\end{proof}

Figure~\ref{fig:universal-section} records both requirements: a continuous
section makes representative projection quotient, and agreement of the
scoped and homotopy relations identifies the quotient spaces.

\MainFigure{universal-section}{A continuous one-letter section splits representative projection; matching scoped and homotopy relations then identify their quotient spaces.}

\begin{theorem}[universal completion]\label{thm:universal-completion}
For the universal presentation $\mathcal U_X$, scoped rewrite equality of
parallel traces is equivalent to endpoint-fixed homotopy of their realizations.
Consequently,
\[
  G_{\mathcal U_X}\cong \Pi_1^q(X)
\]
as abstract groupoids, and the comparison is a homeomorphism on arrow spaces.
\end{theorem}

\begin{proof}
Soundness is the general soundness theorem.  Conversely, an endpoint-fixed
homotopy is itself one of the named universal generators, so it is a scoped
rewrite in one step.  The equivalence of quotient relations follows after
including endpoint equality.  The coherent-path quotient identification lemma
supplies the missing topological identification with $\Pi_1^q(X)$, and the
geometric completeness theorem then gives the homeomorphism and groupoid
isomorphism.
\end{proof}

\begin{corollary}[collapse of the two universal topologies]
\label{cor:universal-topology-collapse}
For the universal presentation, the comparison map from
Proposition~\ref{prop:trace-sensitive-refinement}
\[
  J_{\mathcal U_X}:G_{\mathcal U_X}^{\mathrm{tr}}
    \longrightarrow G_{\mathcal U_X}^{\mathrm{obs}}
\]
is a homeomorphism.  More precisely, both spaces are canonically
homeomorphic to $\Pi_1^q(X)$, and these homeomorphisms identify
$J_{\mathcal U_X}$ with the identity on path classes.
\end{corollary}

\begin{proof}
The coherent-path quotient identification lemma applies to both representative
topologies and identifies each universal scoped quotient with
$\Pi_1^q(X)$.  The quotient maps are induced by the same representative
realization $\gamma$, so the two identifications commute with the identity on
path classes.  Hence the comparison map in
Proposition~\ref{prop:trace-sensitive-refinement}
is the composite of one homeomorphism with the inverse of the other, and is
therefore a homeomorphism.
\end{proof}

\begin{proposition}[the universal path projection]
\label{prop:universal-path-projection}
Let $q_I^{\mathcal U}:X^I\to G_{\mathcal U_X}^{\mathrm{obs}}$ send a path to
the class of its one-letter coherent representative.  This map is quotient.
If $q_I^{\mathcal U}$ is open, the universal presentation has
product-quotient compatibility, and multiplication is continuous on the
ordinary pullback of its arrow space.
\end{proposition}

\begin{proof}
The geometric-path projection $r:T_{\mathcal U_X}^{\mathrm{obs}}\to X^I$
is continuous and has the continuous one-letter section $\sigma$.
Every coherent representative has the same scoped class as the one-letter
representative of its geometric path.  Thus
$q_T=q_I^{\mathcal U}\circ r$ and
$q_I^{\mathcal U}=q_T\circ\sigma$.  The first identity and the quotient
property of $q_T$ make $q_I^{\mathcal U}$ quotient.

On composable pairs, $r\times r$ again has a continuous section, so its
restriction to the endpoint pullbacks is quotient.  If
$q_I^{\mathcal U}$ is open, its square remains an open quotient after
restriction to the full inverse image of composable arrow pairs.
The composite is the ordinary pair map for $\mathcal U_X$.
The ordinary/final comparison theorem now gives product-quotient
compatibility and continuity of multiplication.
\end{proof}

\begin{theorem}[a positive geometric class]
\label{thm:universal-open-groupoid}
Let $X$ be locally path-connected and semilocally simply connected.  The
universal observable presentation $\mathcal U_X$ has product-quotient
compatibility.  Its scoped quotient is therefore a topological groupoid for
the ordinary pullback topology.  For every $x\in X$, the based arrow fiber
$G_{\mathcal U_X}(x,x)$ with the subspace topology from the global arrow
space is discrete, and its multiplication is continuous on the ordinary
product.
\end{theorem}

\begin{proof}
Under these hypotheses, the path-class projection
$q_I:X^I\to\Pi_1^q(X)$ is open by
\cite[Corollary~3.7]{HolkarHossainKulkarni24}; their Theorem~3.9 gives
the resulting ordinary topological groupoid. The homeomorphism of
Theorem~\ref{thm:universal-completion} identifies $q_I$ with
$q_I^{\mathcal U}$. Proposition~\ref{prop:universal-path-projection}
therefore gives the claimed product-quotient compatibility and
ordinary composition continuity at the presentation level.

The based loop space is the full inverse image of the based arrow fiber
under the open map $q_I$. Restriction remains open and quotient, so
the fiber subspace has exactly the usual based-loop quotient topology.
The one-letter section gives the same topology from the coherent based
carrier. For discreteness, apply
\cite[Corollary~3.10(iii)]{HolkarHossainKulkarni24} to the path component
of $x$, which is open by local path-connectedness and satisfies the same
local hypotheses. Every based loop and its based homotopies remain in
that component, with the same compact-open topology. Thus its usual
fundamental-group quotient, and hence the displayed fiber, is discrete.
Multiplication on its ordinary product is consequently continuous.
\end{proof}

\begin{corollary}[totally disconnected spaces]
\label{cor:totally-disconnected-universal}
If $X$ is totally disconnected, $q_I^{\mathcal U}$ is a homeomorphism.
The universal presentation has product-quotient compatibility and
continuous multiplication on the ordinary composable-pair space.
This conclusion does not require local path-connectedness.
\end{corollary}

\begin{proof}
The interval is connected, so every continuous path in $X$ is constant.
Two such paths with the same scoped class have the same source and are
therefore equal. Thus $q_I^{\mathcal U}$ is injective. It is quotient by
Proposition~\ref{prop:universal-path-projection}, hence a homeomorphism.
The remaining claims follow from the resulting openness of
$q_I^{\mathcal U}$ and the open-projection criterion.
\end{proof}

The Hawaiian earring in Proposition~\ref{prop:hawaiian-earring} is locally
path-connected and fails the semilocal hypothesis at its common base point.
These examples locate the ordinary-pullback obstruction in geometric
presentations. The groupoid theorem applies the published compact-open quotient
result; our construction supplies the presentation-level comparison.

\subsection{A strict failure: the Hawaiian earring}

Let
\[
  \mathbb H=\bigcup_{n\geq 1}
  \left\{(x,y)\in\mathbb R^2:
  \left(x-\frac1n\right)^2+y^2=\frac1{n^2}\right\}
\]
be the Hawaiian earring, based at $x_0=(0,0)$.  Write
$\Omega=\Omega(\mathbb H,x_0)$ for its based loop space with the compact-open
topology and
\[
  q_\Omega:\Omega\longrightarrow\pi_1^q(\mathbb H,x_0)
\]
for the endpoint-fixed homotopy quotient.  Fabel proved that multiplication
in $\pi_1^q(\mathbb H,x_0)$ is not continuous \cite{Fabel11}.  Loop
concatenation $\Omega\times\Omega\to\Omega$ is continuous and induces this
multiplication, so $q_\Omega\times q_\Omega$ is not a quotient map either.

\begin{proposition}[Hawaiian-earring obstruction]
\label{prop:hawaiian-earring}
Use the observable based fiber of the universal presentation and the
fiberwise quotient construction of Lemma~\ref{lem:fiberwise-coherent-path-quotient}:
\[
  T_0:=T_{\mathcal U_{\mathbb H}}^{\mathrm{obs}}(x_0,x_0),
  \qquad
  G_0:=G_{\mathcal U_{\mathbb H}}^{\mathrm{obs}}(x_0,x_0),
\]
with quotient map $q_0:T_0\to G_0$ and projection
$\pi_0:T_0\to\Omega$.  Then the ordinary
pair map
\[
  q_0\times q_0:T_0\times T_0\longrightarrow G_0\times G_0
\]
is not quotient.  Consequently the final topology on composable pairs is
strictly finer than the ordinary product topology, and multiplication on the
ordinary product is discontinuous.
\end{proposition}

\begin{proof}
The fiberwise lemma gives that $\pi_0$ is quotient and has the continuous
section $\sigma(\gamma)=(\langle\gamma\rangle,\gamma)$.  It also gives a
fiberwise universal-completion homeomorphism
\[
  h:G_0\longrightarrow\pi_1^q(\mathbb H,x_0),
  \qquad h\circ q_0=q_\Omega\circ\pi_0.
\]
Suppose that $q_0\times q_0$ were
quotient.  Postcomposition with the homeomorphism $h\times h$ would make
\[
  (q_\Omega\times q_\Omega)\circ(\pi_0\times\pi_0)
\]
a quotient map.  If $(q_\Omega\times q_\Omega)^{-1}(U)$ were open, its
preimage under the continuous map $\pi_0\times\pi_0$ would therefore be open;
the quotient criterion for the displayed composite would force $U$ to be
open.  Hence $q_\Omega\times q_\Omega$ would be quotient, contradicting
Fabel's theorem.

Theorem~\ref{thm:ordinary-final} and its proof apply to the based fiber
without change, with $T_0\times T_0$ in place of $T^{(2)}$.  They show that the
continuous bijection from the final pair domain to the ordinary one is not a
homeomorphism.  Since the final topology is always at least as fine, it is
strictly finer here.  Finally, $h$ preserves loop multiplication, so
continuity of ordinary-product multiplication on $G_0$ would imply continuity
on $\pi_1^q(\mathbb H,x_0)$, again contradicting Fabel's theorem.
\end{proof}

Proposition~\ref{prop:hawaiian-earring} is stated for the observable quotient
used in Theorem~\ref{thm:ordinary-final}. For this universal presentation,
Corollary~\ref{cor:universal-topology-collapse} also applies on the based
carrier: the one-letter path section is continuous for the trace-sensitive
topology and recovers the same homotopy class. Its product with itself is a
continuous section on based composable representatives. Repeating the
product-quotient argument above with that section transfers the strict
final/ordinary pair distinction and the discontinuity of ordinary
multiplication to the trace-sensitive based quotient. The general refinement
map alone would not establish that conclusion. The pair-section argument
here is mathematical; the formal selection of the universal-collapse module
checks the arrow comparison.

Fabel's discontinuity theorem thus detects the failure of continuity of
$j^{-1}$.

\begin{remark}[Two distinct obligations]\label{rem:two-obligations}
Geometric completeness concerns \emph{which arrows are equal}:
it compares the scoped and geometric equivalence relations.
Product-quotient compatibility concerns \emph{which composable-pair sets
are open}: it compares the final and ordinary domain topologies and permits
continuity of multiplication to pass to the ordinary pullback.
The former does not imply the latter.  The universal presentation is
geometrically complete for every $X$, yet in the Hawaiian-earring example of
Proposition~\ref{prop:hawaiian-earring} multiplication on its observable
based-fiber quotient is discontinuous for the ordinary product topology.
This also obstructs global compatibility. The base point is closed in the
Hausdorff Hawaiian earring, so continuity of source and target makes its
based arrow fiber closed. A quotient map restricted over a closed target
subset remains quotient: a closed restricted preimage is closed in the
whole domain, and the quotient criterion then gives closedness in the
target. Thus the based quotient has the global fiber's subspace topology.
If the global ordinary pair map were quotient, its restriction over the
closed based-pair subset would be quotient as well, contradicting
Proposition~\ref{prop:hawaiian-earring}.
\end{remark}
\section{Complete presentations and their products}
\label{sec:circle-torus}

The applications below use normal forms and covering invariants.
Whenever a based quotient $G_{\mathcal P}(a,a)$ is used, it means the
quotient of $T(a,a)$ with its own quotient topology; no identification
with a global fiber subspace is assumed.

\subsection{Normal-form certificates}
\label{sec:normal-forms}

\begin{definition}[normal-form certificate]
A normal-form certificate for $\mathcal P$ consists of a set
$\mathsf{NF}$ of normal-form codes, a representative map
$\nu:\mathsf{NF}\to T$, and a normalizer $N:T\to\mathsf{NF}$ such that,
for all coherent representatives $u,v$,
\begin{enumerate}[label=(\alph*),leftmargin=2.2em]
\item $u\approx\nu(N(u))$; and
\item $h(u)=h(v)$ implies $N(u)=N(v)$.
\end{enumerate}

Clause (a) requires a scoped derivation to the chosen normal representative.
Clause (b) requires geometrically equivalent representatives to receive the
same code; equivalently, representatives with different codes lie in
different geometric classes.
A code records only the normal form; it stores neither a chosen homotopy nor
a proof of normalization.  No further condition on the representatives is
needed, because equal codes have equal images under $\nu$.

\end{definition}

\begin{theorem}[normal-form completeness criterion]\label{thm:normal-form-completeness}
A normal-form certificate implies geometric completeness.  Therefore
the realization comparison is faithful and is a homeomorphism on arrow spaces.
\end{theorem}

\begin{proof}
Assume $h(u)=h(v)$.  Clause (b) gives
$N(u)=N(v)$.  Combining this with clause (a) for $u$ and the symmetric form
of clause (a) for $v$ yields

\[
  u\approx\nu(N(u))=\nu(N(v))\approx v.
\]
This is geometric completeness, and the faithfulness and homeomorphism
claims follow from the comparison theorem.  Figure~\ref{fig:normal-form}
shows the two obligations separately.  Neither $N$ nor $\nu$ needs to be
continuous.  The certificate is a set-level object, and the topological
conclusion comes from quotient factorization.
\end{proof}

\MainFigure{normal-form}{Syntactic normalization connects a trace to its normal form; geometric separation says homotopic normal forms have the same index.}

\begin{corollary}[fiberwise normal-form completeness]\label{cor:fiberwise-normal-form}
For fixed endpoints $a,b\in X$, the same definitions and proof apply to
$T(a,b)\subseteq T$.
In particular, a normal-form certificate on $T(a,a)$ proves geometric
completeness of the based-loop comparison on that fiber.
\end{corollary}

\begin{proof}
All endpoint maps and the scoped relation preserve fixed endpoint fibers, so
the two normal-form clauses and the proof of the criterion restrict to
$T(a,b)$ without change.
\end{proof}

\subsection{A complete presentation of the typed reduction diamond}
\label{sec:typed-diamond}
The two reduction orders in Section~\ref{sec:introduction} give a small
logical instance of the completeness criterion. Realize $M,N,P,f\,a$
as the four vertices of the unit square, with $M=(0,0)$, $N=(1,0)$,
$P=(0,1)$, and $f\,a=(1,1)$. Realize $r,s,r',s'$ by the corresponding
oriented straight edges, and give these four primitive labels the discrete
topology. Convex interpolation makes the realizations of $r;s$ and
$r';s'$ homotopic with fixed endpoints.

With no named witness equations, signed $r$-count distinguishes these
two traces, so geometric homotopy does not imply scoped equality.
Adjoin the single diamond equation
\[
 r;s\RwEq r';s'.
\]
It derives $s'\RwEq(r')^{-1};r;s$ by composing with $(r')^{-1}$ and
using cancellation. Eliminate $s'$ and its inverse with this equation.
The remaining edges $r,s,r'$ form a tree. Adjacent cancellation reduces
every trace to its unique reduced tree path between its endpoints.
All geometric paths in the square with the same endpoints are homotopic,
so these normal forms satisfy both clauses of the certificate.
The presentation is therefore geometrically complete on the four
displayed objects and represents every path class between them.
This interprets the specified witness equation; it makes no claim
that the full lambda calculus supplies every such equation automatically.

\subsection{The circle from one generator}

Let $\mathbb T=\mathbb R/\mathbb Z$ carry the quotient topology.  The minimal
presentation has one oriented primitive step $a$ at the base point $0$, with
realization
\[
  a(t)=[t]\in\mathbb T.
\]
The formal reversal supplies $a^{-1}$; no integer-indexed primitive step is
assumed.  A based trace is therefore a finite signed word in $a$ and
$a^{-1}$.  No additional cancellation generator is needed: the structural
inverse coherences in Definition~\ref{def:scoped-rewrite-equality} instantiate to
\(a;a^{-1}\RwEq1_0\) and \(a^{-1};a\RwEq1_0\), and are closed under the
scoped congruence.

For a signed word $w$, let $\operatorname{exp}(w)$ be the number of positive
letters minus the number of negative letters.  Let $c_n$ be the canonical
word consisting of $n$ copies of $a$ when $n\geq 0$ and $-n$ copies of
$a^{-1}$ when $n<0$.

\begin{lemma}[finite-generator reduction]
Every signed word $w$ has a finite scoped derivation
\[
  w\RwEq c_{\operatorname{exp}(w)}.
\]
If two canonical words have endpoint-fixed homotopic realizations, then their
exponents are equal.  In particular, scoped-equivalent canonical words have
equal exponents by soundness.
\end{lemma}

\begin{proof}
Scan the word from left to right while retaining a reduced signed word.  If
the next sign agrees with the last retained sign, append it.  If the signs
differ, delete the adjacent inverse pair using the corresponding structural
coherence inside the surrounding prefix and suffix, using concatenation
congruence.  Each deletion strictly decreases the word length, so the
procedure terminates.  A terminal reduced word has no adjacent sign change
and hence consists entirely of positive letters or entirely of negative
letters.  The integer $\operatorname{exp}$ is unchanged by every deletion,
so the terminal word is exactly $c_{\operatorname{exp}(w)}$.  The same
invariant is the endpoint of the lift to $\mathbb R$; therefore endpoint-fixed
homotopy of canonical words forces equality of their exponents.
\end{proof}

For example, the cancellation mechanism gives the concrete scoped reduction
\[
  a;a^{-1};a^{-1};a;a
  \;\RwEq\;a^{-1};a;a
  \;\RwEq\;a.
\]
The first step cancels the initial $a;a^{-1}$ pair and the second cancels the
remaining $a^{-1};a$ pair.

The reduction and lifted endpoint give clauses (a) and (b) of the
normal-form certificate of Section~\ref{sec:normal-forms}.

Write $\lambda_n(t)=[tn]$ for the realized normal form $|c_n|$.
These loops are derived from the single generator, not additional primitives.

Let $L$ be the endpoint-fixed homotopy quotient of the based circle loop
space, with its quotient topology, and denote by $S_0$ the inverse
realization comparison supplied by the following theorem.

\begin{theorem}[circle winding classification]
\label{thm:circle-winding}
The finite-generator scoped presentation is geometrically complete on the
based loop fiber, and winding number is a group isomorphism and a
homeomorphism
\[
  L\cong\Z,
\]
where $\Z$ is discrete.
Moreover, the inverse realization comparison
\[
  S_0:L\xrightarrow{\ \cong\ }G_{\mathcal P}(0,0)
\]
is a homeomorphism.  In particular the scoped presentation is
nondegenerate: the images of the zero and one winding loops are distinct.
\end{theorem}

\begin{proof}
A based loop has a unique lift to $\mathbb R$ starting at zero.
Its integer endpoint $n$ is invariant under endpoint-fixed homotopy.
Straight-line interpolation of the lift with $t\mapsto tn$ gives a
homotopy to $\lambda_n$. The lift of a concatenation follows the first
lift and then a translate of the second, so winding is additive.
Thus winding identifies $L$ with $\Z$ as groups. The reduction lemma
normalizes each trace to $c_n$, with $|c_n|=\lambda_n$; soundness separates
different integers. Hence the presentation is complete and
$S_0([\gamma])=q(c_n,\lambda_n)$ is inverse to realization.

For continuity, use the standard quotient metric on $\mathbb T$.
If based loops $\gamma,\delta$ are uniformly less than $1/2$ apart,
their difference $\delta-\gamma$ has a unique continuous lift
$d:I\to(-1/2,1/2)$ with $d(0)=d(1)=0$.
Then $(u,t)\mapsto\gamma(t)+[u d(t)]$ is an endpoint-fixed homotopy.
The compact-open and uniform topologies agree for these paths, so winding
is locally constant. It descends continuously to $L\to\Z$; the inverse
$n\mapsto[\lambda_n]$ is continuous because $\Z$ is discrete.
The same discreteness makes $n\mapsto q(c_n,\lambda_n)$ continuous.
Projection to the chosen loop induces its continuous inverse, proving
the asserted scoped homeomorphism for either representative topology.

The topology inherited from the global arrow quotient can be checked
separately. Every trace in this presentation has equal endpoints.
On its full coherent carrier define
$W(p,\gamma)=\operatorname{wind}(t\mapsto\gamma(t)-s(p))$.
Translation and winding are continuous, and soundness makes $W$
constant on scoped classes, so it descends continuously to the global
quotient. Its restriction to the based subspace is bijective by the
same normal forms and has continuous inverse from discrete $\Z$.
Thus this actual subspace also has the claimed discrete topology.
\end{proof}

\subsection{The torus and its two generators}

Let
\[
  \mathbb T^2=(\mathbb R/\mathbb Z)\times(\mathbb R/\mathbb Z)
\]
with base point $(0,0)$. Use the discrete primitive alphabet with two
oriented generators
\[
  a(t)=([t],[0]),
  \qquad b(t)=([0],[t]).
\]
The only named rule is the commuting square
\[
  a;b\RwEq b;a;
\]
inverse cancellations for $a$ and $b$ are structural, as for the circle.
The commuting square is sound: it is the projection modulo $\mathbb Z^2$ of
the evident square homotopy in $\mathbb R^2$.  The other signed commutations
are derived, rather than added as generators.  For example, with unit
insertions and structural cancellations displayed schematically,
\[
\begin{aligned}
a^{-1};b
  &\RwEq a^{-1};(b;a);a^{-1}
   \RwEq a^{-1};(a;b);a^{-1}
   \RwEq b;a^{-1},\\
a;b^{-1}
  &\RwEq b^{-1};(b;a);b^{-1}
   \RwEq b^{-1};(a;b);b^{-1}
   \RwEq b^{-1};a,\\
b^{-1};a^{-1}
  &\RwEq a^{-1};b^{-1}
\end{aligned}
\]
and symmetry gives the reversed orientations.  Thus closure under the
scoped congruence supplies all four signed commuting rules.

For a signed word $w$, let $m(w)$ and $n(w)$ be its signed counts of $a$ and
$b$.  Move every $b$-letter to the right using the commuting square and then
perform the two cancellation algorithms from the circle.  This is a finite
terminating derivation
\[
  w\RwEq a^{m(w)};b^{n(w)}.
\]
The pair $(m(w),n(w))$ is invariant under every scoped rule.
Each signed commutation strictly decreases the number of pairs in which a
$b^{\pm1}$-letter occurs to the left of an $a^{\pm1}$-letter.
For a short worked instance, sorting and cancellation give
\[
  b;a;b^{-1};a;b
  \;\RwEq\;a;b;b^{-1};a;b
    \;\RwEq\;a;a;b
  \;=\;a^2;b.
\]
The first step uses the commuting square and the second cancels $b;b^{-1}$;
the winding pair is $(2,1)$.

\begin{lemma}[simultaneous and sequential torus representatives]
\label{lem:torus-sequential}
Let
\[
  d_{m,n}(t)=([mt],[nt])
\]
be the simultaneous standard loop, and let $|a^m;b^n|$ be the equal-slot
realization of the sequential word.  Then
\[
  d_{m,n}\simeq_{\partial I}|a^m;b^n|.
\]
\end{lemma}

\begin{proof}
Choose lifts of both loops to paths in $\mathbb R^2$ beginning at $(0,0)$.
The lift of $d_{m,n}$ is $t\mapsto(mt,nt)$, while the lift of the sequential
equal-slot realization is piecewise linear.  Both lifts end at $(m,n)$.
The straight-line interpolation between these two lifted paths is therefore
fixed at both endpoints.  Composing it with the covering map
$\mathbb R^2\to\mathbb T^2$ gives the required endpoint-fixed homotopy.
\end{proof}

Let $L_{\mathbb T^2}$ be the endpoint-fixed homotopy quotient of the based
loop space of $\mathbb T^2$, with its quotient topology, and denote by
$S_{\mathbb T^2}$ the inverse realization comparison supplied by the
following theorem.

\begin{theorem}[torus winding classification]
\label{thm:torus-winding}
For the torus $\mathbb T^2$, the coordinate winding pair is a group
isomorphism and a homeomorphism
\[
  L_{\mathbb T^2}\cong\Z\times\Z,
\]
where $\Z\times\Z$ is discrete.
The finite-generator scoped presentation is geometrically complete on the
based loop fiber.  Moreover, the inverse realization comparison
\[
  S_{\mathbb T^2}:L_{\mathbb T^2}\xrightarrow{\ \cong\ }
  G_{\mathcal P}((0,0),(0,0))
\]
is a homeomorphism.
\end{theorem}

\begin{proof}
Lift a based loop coordinatewise to $\mathbb R^2$ starting at $(0,0)$.
Its endpoint $(m,n)\in\Z^2$ is homotopy invariant, and interpolation
with $t\mapsto(mt,nt)$ gives a homotopy to $d_{m,n}$.
Coordinate winding is additive, so it identifies $L_{\mathbb T^2}$
with $\Z^2$ as groups. The signed commuting rules and cancellation
normalize every word to $a^m;b^n$; Lemma~\ref{lem:torus-sequential}
makes $(a^m;b^n,d_{m,n})$ coherent. Equal geometric winding pairs
therefore give a scoped derivation through this common normal form,
proving completeness.

Each coordinate winding is locally constant by the circle argument.
The winding pair consequently descends continuously to
$L_{\mathbb T^2}\to\Z^2$, whose inverse sends $(m,n)$ to $[d_{m,n}]$.
Both this inverse and $(m,n)\mapsto q(a^m;b^n,d_{m,n})$ are continuous
from discrete $\Z^2$. Projection to the chosen loop gives the continuous
realization comparison, proving the stated homeomorphisms.
For the actual based subspace of the global arrow quotient, take the
two windings of $t\mapsto\gamma(t)-s(p)$ on the full coherent carrier.
As for the circle, continuity and soundness give a continuous descended
map to $\Z^2$. On the based subspace it is bijective by normalization,
and its inverse is continuous from the discrete domain.
\end{proof}

Figure~\ref{fig:winding} illustrates the lifted paths used in both proofs.
\MainFigure{winding}{A lifted circle loop ends at an integer; coordinate lifts of a torus loop end at an integer pair, determining their normal forms.}

\subsection{Products of complete based presentations}

The torus sorting argument extends to a product without requiring normal
forms for either factor.  Let $\mathcal P$ be a sound scoped presentation
on $X$ with primitive steps $E$, and let $\mathcal Q$ be one on $Y$ with
primitive steps $F$. Give the product primitive space
$(E\times Y)\sqcup(X\times F)$ its product and coproduct topologies.
Its horizontal step $(e,y')$ has endpoints $(s(e),y')$, $(t(e),y')$
and realization $t\mapsto(\rho(e)(t),y')$; a vertical step $(x',f)$
has endpoints $(x',s(f))$, $(x',t(f))$ and realization
$t\mapsto(x',\rho(f)(t))$. These endpoint maps are continuous, and
compact-open evaluation and the exponential law make both realization
maps continuous. Lift each
named rule of $\mathcal P$ at every $y'$, and each rule of $\mathcal Q$ at
every $x'$.  Add interchange around every primitive rectangle:
\[
 (e,y_0);(x_1,f)\ \simeq_{\mathcal P\boxtimes\mathcal Q}\
 (x_0,f);(e,y_1),
 \quad e:x_0\to x_1,\quad f:y_0\to y_1.
\]
The corresponding interchange for arbitrary factor traces $p:x_0\to x_1$
and $q:y_0\to y_1$ follows from these primitive squares:
\[
 H_{y_0}(p);V_{x_1}(q)\ \simeq_{\mathcal P\boxtimes\mathcal Q}\
 V_{x_0}(q);H_{y_1}(p).
\]
Here $H_y$ and $V_x$ lift entire traces while holding the other coordinate
fixed. To derive the displayed relation, fix a vertical primitive step
$f$ and induct on $p$. The single-step case is the named rectangle rule;
the unit case uses the groupoid unit laws. For a concatenation, paste the
two previously derived squares along their common edge. For an inverse,
reverse a square and cancel the resulting inverse pairs. Induction on $q$
then gives the general trace relation by the same pasting and inverse steps.
Write $\mathcal P\boxtimes\mathcal Q$ for the presentation with primitive
rectangle rules. Its rules are sound: the two paths are the two edge routes from $(x_0,y_0)$
to $(x_1,y_1)$ in the square
$(u,v)\mapsto(\rho(e)(u),\rho(f)(v))$.  The routes are homotopic relative
to their endpoints inside the square.

\begin{theorem}[product closure of based completeness]
\label{thm:product-completeness}
Suppose $\mathcal P$ is geometrically complete on the based-loop fiber at
$x\in X$ and $\mathcal Q$ is geometrically complete on the based-loop
fiber at $y\in Y$.  Then $\mathcal P\boxtimes\mathcal Q$ is geometrically
complete on the based-loop fiber at $(x,y)$.
\end{theorem}

\begin{proof}
The derived interchange moves a vertical factor word past a horizontal
one. Induction on concatenation and reversal therefore rewrites each
product trace into its horizontal projection followed by its vertical
projection. For a trace based at $(x,y)$ the result is
$H_y(p);V_x(q)$, where $p$ is a loop at $x$ and $q$ a loop at $y$.

Suppose two coherent product loops have homotopic chosen representatives.
Coherence and soundness of sorting give a homotopy between their sorted
realizations. Project it to $X$ and $Y$. In either projection the letters
of the other factor contribute only constant pauses. Deleting these
pauses gives the corresponding factor realization up to endpoint-fixed
homotopy: interpolate its nondecreasing piecewise-linear time parameter
with the identity; an empty factor is constant. Thus the two horizontal
factor realizations are homotopic, as are the two vertical ones.
Factor completeness supplies scoped derivations between the respective
factor words. Lift the horizontal derivation at $y$ and the vertical
derivation at $x$, then compose with the sorting derivations. This gives
the required scoped equality of the original product traces; the converse
is soundness.
\end{proof}

Iterating the theorem from the finite oriented circle presentation gives
completeness of the corresponding product presentation of every
$n$-torus on its based-loop fiber. The product construction includes a
horizontal or vertical copy of each factor step at every point of the other
factor, so its step space is generally infinite. It is distinct from the
finite-generator torus presentation of Theorem~\ref{thm:torus-winding}.
The Lean development checks sorting and based trace completeness for a
presentation that names whole-trace interchange. It also derives every
such rule from primitive squares, proves equivalence of the two scoped
rewrite relations, and transfers completeness to the primitive-only
presentation.

\subsection{A geometric nonabelian example: the two-circle rose}
Figure~\ref{fig:rose-covering} shows how the covering retains word order.

\label{sec:rose}

Let $R$ be the quotient of two copies of $\mathbb R/\mathbb Z$ identifying
their zero points, with basepoint $o$.  It is homeomorphic to the two
coordinate circles inside $(\mathbb R/\mathbb Z)^2$.  Let $a$ and $b$ be
their standard once-around based loops.  Give the two primitive labels the
discrete topology and declare no named rewrites.  Thus cancellation and the
structural groupoid rules are the entire scoped presentation.

Write $F(a,b)$ for the free group on the two labels.  Reading a signed trace
as a group word gives $\operatorname{ev}(p)\in F(a,b)$.  Its reduced word gives
a normal form: structural scoped rewriting relates $p$ to the trace
reconstructed from that word.  This algebraic normalization
works for any two based primitive loops.  The following geometric bridge
uses the particular space $R$.

\begin{theorem}[represented-loop completeness for the rose]
\label{thm:rose-completeness}
For based traces $p,q$ in this presentation,
\[
  p\simeq_{\mathcal P}q
  \quad\Longleftrightarrow\quad
  \operatorname{ev}(p)=\operatorname{ev}(q)
  \quad\Longleftrightarrow\quad
  |p|\simeq_{\partial I}|q|.
\]
Both the trace-sensitive and equal-slot observable fixed-endpoint based
scoped quotients are discrete.  In particular, $a;b$ and $b;a$ are
geometrically nonhomotopic.
\end{theorem}

\begin{proof}
The normal-form argument gives the first equivalence, and soundness gives
the forward geometric implication.  For the converse, construct a space
$\widetilde R$ with vertices $g\in F(a,b)$ and one closed interval from $g$ to $gs$
for each $s\in\{a,b\}$.  Its topology is the final topology of these
closed intervals, with precisely the specified endpoint identifications.
Map each interval to the corresponding primitive circle loop and each
vertex to $o$.

\begin{figure}[htbp]
\centering
\begin{tikzpicture}[scale=1, every node/.style={font=\small}]
  \draw (-4.8,0) circle (.65);
  \draw (-3.5,0) circle (.65);
  \fill (-4.15,0) circle (2pt);
  \node[below=4pt] at (-4.15,0) {$o$};
  \node[above] at (-4.8,.65) {$a$};
  \node[above] at (-3.5,.65) {$b$};
  \node[below] at (-4.15,-1.45) {$R$};
  \draw[->,thick] (-.8,0) -- (-2.5,0) node[midway,above] {$p$};
  \coordinate (one) at (1,0);
  \coordinate (a) at (0,1);
  \coordinate (b) at (2,1);
  \coordinate (ai) at (0,-1);
  \coordinate (bi) at (2,-1);
  \coordinate (ab) at (-.55,2);
  \coordinate (ba) at (2.55,2);
  \draw[->] (one) -- (a) node[midway,above right] {$a$};
  \draw[->] (one) -- (b) node[midway,above left] {$b$};
  \draw[->] (ai) -- (one) node[midway,above left] {$a$};
  \draw[->] (bi) -- (one) node[midway,above right] {$b$};
  \draw[->] (a) -- (ab) node[midway,right] {$b$};
  \draw[->] (b) -- (ba) node[midway,left] {$a$};
  \foreach \v in {one,a,b,ai,bi,ab,ba} \fill (\v) circle (2pt);
  \node[below] at (one) {$1$};
  \node[above left=2pt] at (a) {$a$};
  \node[above right=2pt] at (b) {$b$};
  \node[below left] at (ai) {$a^{-1}$};
  \node[below right] at (bi) {$b^{-1}$};
  \node[above] at (ab) {$ab$};
  \node[above] at (ba) {$ba$};
  \node[below] at (1,-1.45) {part of $\widetilde R$};
\end{tikzpicture}
\caption{The geometric covering records word order in lifted endpoints.
The two words $ab$ and $ba$ end at different vertices above the same
basepoint.  The pictured finite part belongs to the full interval-glued
Cayley realization.}
\FigureAltText{The two-circle rose is covered by an interval-glued Cayley graph. Lifts of ab and ba end at distinct vertices above its basepoint.}
\label{fig:rose-covering}
\end{figure}

Left multiplication acts continuously on $\widetilde R$.  The map to $R$ is a
quotient map: the finite union of the two unit intervals maps continuously
and surjectively from a compact space to the Hausdorff rose, and the
projection forgetting the free-group coordinate is open and surjective.
Descent through the interval-gluing quotient preserves this conclusion.
The fibers are exactly the left-multiplication orbits.  This follows from
the injective open-interval parametrization of each punctured circle and
the fact that the two circles meet only at $o$.

The action has disjoint local translates.  Interior points have open
neighborhoods contained in a single edge.  At a vertex $v$, extend the
word distance $\|v^{-1}g\|$ affinely along each edge.  This is a continuous
function $d_v$ on $\widetilde R$.  The triangle inequality for reduced-word length
gives
\[
  \|v^{-1}w\|\le d_v(z)+d_w(z).
\]
Consequently the open sets $d_v<1/3$ and $d_w<1/3$ are disjoint when
$v\ne w$, and translation carries one such neighborhood to the
corresponding translated vertex neighborhood.  The orbit map is therefore
a covering map.

A finite signed word has an explicit lift starting at the identity
vertex, obtained by traversing the corresponding oriented edges.  Its
endpoint is the vertex given by its free-group evaluation.  The
equal-slot realization is homotopic to the recursively concatenated
realization, so homotopic trace realizations have lifts with the same
endpoint.  Vertex injectivity now gives
$\operatorname{ev}(p)=\operatorname{ev}(q)$ and hence scoped equality.
The two reduced words $ab$ and $ba$ differ.

The full-word based quotient is discrete because the primitive alphabet
is discrete.  For the observable quotient, the four signed primitive
paths are distinct: the two circle retractions and their winding numbers
distinguish them.  Equal-slot restriction recovers each signed letter.
Each fixed-length observation stratum is finite inside a Hausdorff path
space, hence discrete, and the discrete length coordinate separates the
strata.  Coherent-representative elimination identifies the coherent scoped
quotients homeomorphically with these trace-only quotients, so they too are
discrete.
\end{proof}

This proof is a formalized instance of the classical graph-covering
method~\cite[Section~1.A]{Hatcher02}.  The next construction establishes
the simple connectivity needed to classify every continuous based loop.

\begin{samepage}
\begin{theorem}[all continuous loops of the rose]
\label{thm:rose-all-loops}
The geometric Cayley tree $\widetilde R$ is contractible. Every continuous based loop
in $R$ is endpoint-fixed homotopic to the equal-slot realization of a
reduced finite word in $a^{\pm1},b^{\pm1}$. Two loops are homotopic exactly
when their assigned reduced words agree. The resulting realization
comparison is a bijection from the scoped based quotient onto $\pi_1(R)$.
For either trace-sensitive or equal-slot observable topology, this
bijection is a homeomorphism onto the compact-open homotopy quotient
$\pi_1^q(R)$, and these spaces are discrete.
\end{theorem}
\end{samepage}
\begin{proof}
Write $d(g)$ for reduced-word length and $P_g$ for the equal-slot path
from the identity vertex to $g$, with $P_1$ constant. Adjacent vertices
have lengths differing by one, and the shorter reduced word is an initial
segment of the longer. On an edge from $g$ to $r=ga_i$, with coordinate
$t\in[0,1]$, define
\[
 H(s,t)=
 \begin{cases}
 P_r\!\left(\dfrac{(1-s)(d(g)+t)}{d(g)+1}\right),&d(g)<d(r),\\[2mm]
 P_g\!\left(\dfrac{(1-s)(d(r)+1-t)}{d(r)+1}\right),&d(r)<d(g).
 \end{cases}
\]
At a vertex $g$ this gives $P_g(1-s)$. The equal-slot prefix identities
make the formulas agree at both ends of every edge, including the edges
incident to the identity. Each edge formula is jointly continuous.
Multiplication by the locally compact Hausdorff interval preserves the
closed-edge quotient map, so these formulas descend to a continuous
homotopy $H:I\times\widetilde R\to\widetilde R$. At time zero it is the identity, at time
one it is constant at the identity vertex, and that vertex stays fixed
throughout. Thus $\widetilde R$ is contractible.

Let $p:\widetilde R\to R$ be the covering constructed in
Theorem~\ref{thm:rose-completeness}, and let $\gamma$ be any continuous
loop at $o$. Its unique lift $\widetilde\gamma$ starting at the identity
vertex ends at a vertex $g\in F(a,b)$. The loop
$\widetilde\gamma;P_g^{-1}$ contracts in $\widetilde R$ by the homotopy above.
The homotopy groupoid laws therefore give an endpoint-fixed homotopy
$\widetilde\gamma\simeq_{\partial I}P_g$. Projecting along $p$ shows
that $\gamma$ is homotopic relative to its endpoints to the equal-slot
realization of the reduced word for $g$.

Homotopic based loops have lifts with the same endpoint by covering
homotopy lifting. Conversely, two loops whose lifts end at the same vertex
are both homotopic to the realized reduced word of that vertex. This proves
the classification of all continuous based loops. Combined with
Theorem~\ref{thm:rose-completeness}, it gives the asserted bijection from
the scoped based quotient. Flat and binary realizations are
endpoint-fixed homotopic, so either timing convention gives the same
classification.

The rose has contractible small neighborhoods at the vertex and at
interior edge points, so it is locally path-connected and semilocally
simply connected. Its compact-open fundamental-group quotient is
discrete by Theorem~\ref{thm:universal-open-groupoid}. Both scoped based
quotients are discrete by Theorem~\ref{thm:rose-completeness}. Their
realization bijections are therefore homeomorphisms.
\end{proof}

Mathlib composes the second path first in its fundamental-group
multiplication. Accordingly, the checked endpoint equivalence has target
$F(a,b)^{\mathrm{op}}$; composing with inversion gives a group equivalence
with $F(a,b)$. The direct realization comparison reverses multiplication
under this convention. The checked declarations establish this algebraic comparison. The
last paragraph of the proof supplies the topological conclusion by
combining discreteness with the published geometric criterion.

\section{Checked scope and reproducibility}
\label{sec:formalization}
The proof source is frozen at
\href{https://github.com/Arthur742Ramos/TopologicalComputationalPaths/tree/ea1367d99c7d2f3794f1299cce5b18379e9093a7}{\texttt{ea1367d}}
of the accompanying repository, using Lean~4.32.0 (kernel commit
\texttt{8c9756b}) and Mathlib \texttt{81a5d257}.
The \href{https://github.com/Arthur742Ramos/TopologicalComputationalPaths/actions/runs/38574024275}{verification run for this exact commit}
completed the project build and all five retained Comparator/NanoDa checks.
A separate direct-replay manifest selects 178 actual declarations (127
theorems and 51 definitions) whose 20,013-declaration dependency closure
was accepted by NanoDa. Its axiom audit contains only
\LeanName{propext}, \LeanName{Classical.choice}, and \LeanName{Quot.sound}.
The manifests, inspection commands and retained logs identify the checked
declarations and their assumptions.

The checked equal-slot development includes literal-word correspondence,
the tuple-stratum topology, fixed-endpoint coherent-representative
elimination, literal-word global arrow and final-pair homeomorphisms, variable-length word-code invariance,
conservative-abbreviation discontinuity, and the circle timing
nonhomeomorphism. The complete discrete-label separation theorem assumes
nondiscreteness of the ordinary based quotient; the harmonic-archipelago
and Hawaiian-earring results used to instantiate geometric obstructions
remain cited external inputs. The rose development classifies all continuous
based loops algebraically, with Mathlib's opposite multiplication convention
explicit; its checked topological identifications concern the computational
based quotients.

The general operation-continuity and final-domain groupoid declarations
are checked for binary realization. Their equal-slot counterparts are
proved here: the checked flat bridges identify arrow and pair spaces, but
do not separately verify continuity of the canonical flat operations or
compatibility of those operations with the bridges. The point-model
classification and abbreviation separation, global and pair coherent
elimination, compatibility transfer, finite typed-diamond completeness, and
topological comparison with the rose's ordinary fundamental-group quotient
are additional paper proofs. The accompanying verification map specifies
the formalization boundary for each added conclusion.

\section{Conclusion}
Continuous primitive word codes and scoped roundtrips preserve the
trace-sensitive quotient, its final composable-pair space, and
product-quotient compatibility. Exact observation behaves differently:
even a conservative abbreviation can change a discrete two-element
quotient into an indiscrete one. The classification over a point determines
this loss precisely and shows that it can occur while all finite products
of the arrow quotient remain quotient. Timing and geometric-completeness
examples expose two further sources of topological dependence.

The final composable-pair topology records the continuous multiplication
available directly from representatives. Agreement with the ordinary
pullback is an additional quotient property. An open arrow quotient guarantees
this property, while the Hawaiian-earring example shows that it can fail.
Normal forms and covering geometry show how declared witness equations recover complete geometric
semantics in concrete logical and topological presentations. These results
separate the effects of rewriting, representation, and composition without
identifying one kind of equality or topology with another by assumption.

\section*{Funding}
No external funding was received for this work.

\section*{Acknowledgements}
AI-assisted tools were used in developing and checking mathematical
arguments, revising the exposition, and preparing software, documentation,
and figures. The authors are responsible for the mathematical claims,
source verification, and final manuscript. Mechanically checked claims
are limited to the scope stated in Section~\ref{sec:formalization}.
The geometric exposition and figures were developed from Tiago M.~L.~de
Veras's suggestions on reversal, coherence, and homotopy obstructions.
The figures are mathematical schematics with explicitly stated maps and
boundary conditions. They illustrate the arguments; the proofs establish
the theorems.

\bibliographystyle{amsplain}
\begingroup\small\bibliography{refs}\endgroup

\providecommand{\bysame}{\leavevmode\hbox to3em{\hrulefill}\thinspace}
\providecommand{\MR}{\relax\ifhmode\unskip\space\fi MR }
\providecommand{\MRhref}[2]{%
  \href{http://www.ams.org/mathscinet-getitem?mr=#1}{#2}
}
\providecommand{\href}[2]{#2}
\begin{thebibliography}{10}

\bibitem{Brazas14}
Jeremy Brazas, \emph{Open subgroups of free topological groups}, Fundamenta
  Mathematicae \textbf{226} (2014), no.~1, 17--40.

\bibitem{BrazasEmery23}
Jeremy Brazas and Sarah Emery, \emph{Free quasitopological groups}, Topology
  and its Applications \textbf{326} (2023), 108416.

\bibitem{BrownHardy76}
Ronald Brown and J.~P.~L. Hardy, \emph{Topological groupoids: {I}. universal
  constructions}, Mathematische Nachrichten \textbf{71} (1976), 273--286.

\bibitem{RamosDeQueirozDeOliveira2016propositional}
Ruy J. G.~B. de~Queiroz, Anjolina~Grisi de~Oliveira, and Arthur~Freitas Ramos,
  \emph{Propositional equality, identity types, and computational paths}, South
  American Journal of Logic \textbf{2} (2016), no.~2, 245--296.

\bibitem{DeVerasRamosDeQueirozDeOliveira2025topological}
Tiago M.~L. de~Veras, Arthur~F. Ramos, Ruy J. G.~B. de~Queiroz, and Anjolina~G.
  de~Oliveira, \emph{A topological application of labelled natural deduction},
  2019, \href{https://arxiv.org/abs/1906.09105}{arXiv:1906.09105}, version 4,
  revised 9 May 2021.

\bibitem{DeVerasRamosDeQueirozDeOliveira2025weak}
\bysame, \emph{Computational paths -- a weak groupoid}, Journal of Logic and
  Computation \textbf{35} (2025), no.~5, exad071, Published online 24 November
  2023; issue dated 2025.

\bibitem{Fabel05Archipelago}
Paul Fabel, \emph{The fundamental group of the harmonic archipelago}, 2005,
  \href{https://arxiv.org/abs/math/0501426}{arXiv:math/0501426}.

\bibitem{Fabel11}
\bysame, \emph{Multiplication is discontinuous in the {Hawaiian Earring Group}
  (with the quotient topology)}, Bulletin of the Polish Academy of Sciences.
  Mathematics \textbf{59} (2011), no.~1, 77--83.

\bibitem{HardyMorris77}
J.~P.~L. Hardy and Sidney~A. Morris, \emph{A note on free topological
  groupoids}, Mathematische Nachrichten \textbf{77} (1977), no.~1, 177--180.

\bibitem{Hatcher02}
Allen Hatcher, \emph{Algebraic topology}, Cambridge University Press,
  Cambridge, 2002.

\bibitem{HenryMimram21}
Simon Henry and Samuel Mimram, \emph{Tietze equivalences as weak equivalences},
  Applied Categorical Structures \textbf{30} (2022), no.~3, 453--483,
  \href{https://doi.org/10.1007/s10485-021-09662-w}{doi:10.1007/s10485-021-09662-w}.

\bibitem{HolkarHossainKulkarni24}
Rohit~Dilip Holkar, Md~Amir Hossain, and Dheeraj Kulkarni, \emph{The
  fundamental groupoid: Topology, haar system and actions}, M\"unster Journal
  of Mathematics \textbf{17} (2024), 203--240.

\bibitem{May99}
J.~Peter May, \emph{A concise course in algebraic topology}, Chicago Lectures
  in Mathematics, University of Chicago Press, 1999.

\bibitem{MimramOleon25}
Samuel Mimram and {\'E}mile Oleon, \emph{Coherent {Tietze} transformations of
  {$1$}-polygraphs in homotopy type theory}, 10th International Conference on
  Formal Structures for Computation and Deduction ({FSCD} 2025) (Maribel
  Fern{\'a}ndez, ed.), Leibniz International Proceedings in Informatics
  (LIPIcs), vol. 337, Schloss Dagstuhl -- Leibniz-Zentrum f{\"u}r Informatik,
  2025,
  \href{https://doi.org/10.4230/LIPIcs.FSCD.2025.30}{doi:10.4230/LIPIcs.FSCD.2025.30},
  pp.~30:1--30:17.

\bibitem{Ramos2017identity}
Arthur~Freitas Ramos, Ruy J. G.~B. de~Queiroz, and Anjolina~G. de~Oliveira,
  \emph{On the identity type as the type of computational paths}, Logic Journal
  of the IGPL \textbf{25} (2017), no.~4, 562--584.

\bibitem{RamosDeQueirozDeOliveiraGabbay2026}
Arthur~Freitas Ramos, Ruy J. G.~B. de~Queiroz, Anjolina~Grisi de~Oliveira, and
  Dov~M. Gabbay, \emph{Computational paths: The calculus of equality}, College
  Publications, 2026.

\end{thebibliography}
\end{document}